\documentclass{amsproc}
\usepackage[utf8]{inputenc}
\usepackage{amssymb,amsmath,amsthm,amscd,slashed,mathrsfs}
\usepackage{graphicx,hyperref}

\numberwithin{equation}{section}

\newtheorem{Theorem}{Theorem}[section]
\newtheorem{Proposition}[Theorem]{Proposition}
\newtheorem{Corollary}[Theorem]{Corollary}

\newtheorem{Claim}[Theorem]{Claim}

\newcommand{\id}{\mbox{id}}

\newcommand{\F}{\mathscr{F}}
\newcommand{\A}{\mathscr{A}}
\newcommand{\Ho}{\mathcal{H}}
\newcommand{\hs}[1]{\hspace{#1cm} }
\newcommand{\tree}[5]{\hspace{#1cm} \raisebox{#2cm}{\includegraphics[scale=#3]{#4-crop.pdf}} \hspace{#5cm} }
\newcommand{\graph}[5]{\hspace{#1cm} \raisebox{#2cm}{\includegraphics[scale=#3]{#4.pdf}} \hspace{#5cm} }

\begin{document}

\title{Avoidance of a Landau Pole by Flat Contributions in QED}

\author{Lutz Klaczynski}
\address{Department of Physics, Humboldt University Berlin, 12489 Berlin}
\email{klacz@mathematik.hu-berlin.de}

\author{Dirk Kreimer}
\address{Alexander von Humboldt Chair in Mathematical Physics, Humboldt University Berlin, 12489 Berlin}
\thanks{The second author is supported by the Alexander von Humboldt Foundation and the BMBF}
\email{kreimer@mathematik.hu-berlin.de}

\begin{abstract}
We consider massless Quantum Electrodynamics in momentum scheme and carry forward an approach based on Dyson--Schwinger equations to 
approximate both the $\beta$-function and the renormalized photon self-energy \cite{Y11}. Starting from the Callan-Symanzik equation, 
we derive a renormalization group (RG) recursion identity which implies a non-linear ODE for the anomalous dimension and extract a 
sufficient but not necessary criterion for the existence of a Landau pole. This criterion implies a necessary condition for QED to 
have no such pole. Solving the differential equation exactly for a toy model case, we integrate the corresponding 
RG equation for the running coupling and find that even though the $\beta$-function entails a Landau pole it exhibits a flat 
contribution capable of decreasing its growth, in other cases possibly to the extent that such a pole is avoided altogether. 
Finally, by applying the recursion identity, we compute the photon propagator and investigate the effect of flat contributions on 
both spacelike and timelike photons. 
\end{abstract}

\maketitle

\section{Introduction: non-perturbative contributions}

We shall use the customary fine-structure constant $\alpha$ as coupling parameter. Consider the transversal part of the full renormalized
photon propagator
\begin{equation}\label{Pro}
\Pi_{\mu \nu}(q) = \frac{g_{\mu \nu} - q_{\mu}q_{\nu}/q^{2}}{q^{2}[1-\Sigma(\alpha,-q^{2}/\mu^{2})]} \ ,
\end{equation}
in massless QED with Minkowski metric $g_{\mu \nu}$ in mostly minus signature and renormalization point $\mu^{2}$ in momentum subtraction 
scheme. Let $L=\ln(-q^{2}/\mu^{2})$ be the Minkowksi momentum parameter and 
\begin{equation}\label{Green}
G(\alpha,L)=1-\Sigma(\alpha,e^{L})= 1 - \sum_{k=1}^{\infty}\gamma_{k}(\alpha)L^{k},
\end{equation}
the form factor in the denominator of (\ref{Pro}) which we shall refer to as \emph{Green's function} (of the photon) henceforth.
This function satisfies the renormalization condition $G(\alpha,0)=1$ at the subtraction point $q^{2}=-\mu^{2}$.
It is the object of this publication to find an approximation to this function as well as the corresponding $\beta$-function defined by
\begin{equation}
 \beta(\alpha):=\alpha \gamma(\alpha),
\end{equation}
where $\gamma(\alpha):=\gamma_{1}(\alpha)$ is  the first log-coefficient function in (\ref{Green}) known as \emph{anomalous dimension},
and discuss how instantonic contributions (see below) may help avert a Landau pole by sufficiently hampering its growth.

We briefly describe how this paper is organized and thereby sketch the line of argument for our approach. First, we will review and for 
clarity rederive an earlier
result by \cite{Y11} in Section \ref{secRG} that the \emph{Callan-Symanzik equation} implies the singular non-linear ODE
\begin{equation}\label{ODE1}
 \gamma(\alpha) + \gamma(\alpha)(1-\alpha \partial_{\alpha})\gamma(\alpha) = P(\alpha)
\end{equation}
for the anomalous dimension $\gamma(\alpha)$. We will deduce a sufficient but not necessary criterion for the existence of a Landau pole
from this equation which is equivalent to a \emph{necessary condition for QED to be free of such pole}. Our approach is to approximate $\gamma(\alpha)$ by using a perturbative expansion of
the function $P(\alpha)$ in the coupling $\alpha$. Section 2 explores the perturbative expansion of $P(\alpha)$ in terms of Feynman
diagrams where it becomes obvious that the power of the coupling $\alpha$ signifies the loop order. For a derivation of Eq.(\ref{ODE1}) 
from Dyson-Schwinger equations and their underlying structure, see \cite{Y11,Y13} and the references therein.

In this paper, however, we restrict ourselves to first order perturbation theory, where we 
approximate $P(\alpha) = c \alpha$ with $c>0$ and find a family of 'toy model' solutions 
\begin{equation}\label{LamberW}
 \gamma(\alpha)=c \alpha [1+W(\xi e^{-\frac{1}{c \alpha}})] ,
\end{equation}
indexed by $\xi$, a parameter which is determined by the initial condition for $\gamma(\alpha)$. The famous Lambert W function $W(x)$ 
constitutes the non-perturbative \emph{flat component} (see also \cite{BKUY09}).    
Just to remind the reader, a smooth function $f:(0,\infty) \rightarrow \mathbb{R}$ is called \emph{flat}
if it has a vanishing Taylor series at zero. 
The prime example well known to physicists is $f(\alpha)=\exp(-1/\alpha)$. Flat contributions of this type are in physics usually
referred to as \emph{instantonic}. Section \ref{secFlatCon} explores the mathematical intricacies of 'flat contributions' and what their 
properties have in store for the anomalous dimension as a solution of (\ref{ODE1}).  

Because the flat part we are interested in only comes to life when the coupling is large, we will see in Section \ref{secFoApp} that we 
cannot seriously treat the solution (\ref{LamberW}) as an approximation. Instead, we take it as an interesting toy model with a 
non-perturbative feature that is worthwhile studying for the following reason: it serves as
a guideline along which one can nicely investigate the impact of flat contributions. Two salient aspects are: 
\begin{enumerate}
 \item Though instantonic contributions never appear in standard perturbation theory due to their vanishing Taylor series, they may still
control the asymptotic behaviour of the $\beta$-function in such a way that the running coupling is free of a Landau pole. 
\item Both (toy model) Green's and $\beta$-function are, courtesy of the instantonic contribution, \emph{non-analytic}. According
to Dyson \cite{Dys51}, a feature to be expected for QED. 
\end{enumerate}
We quickly remind the reader of the argument Dyson put forward in favour of non-analyticity at the time \cite{Dys51}: 
if the Green's function were analytic, its Taylor expansion would also converge at some negative value of the coupling parameter. 
This value corresponds to a fictitious world in which the Coulomb potential would (in the classical limit) be attractive for like charges. In this 
scenario, he argues, there is a non-vanishing quantum mechanical chance for a 'pathological' state brought about by spontaneous 
particle-antiparticle creation (in the presence of an external field, that is). Because one can easily imagine a situation in which an 
increasing number of electrons accumulate in one region of space and their positron partners in another, a total disintegration of 
the vacuum is inevitable. This will then backfire on 
the Green's function, forcing it to diverge. Therefore, it must have poles on the negative real axis and cannot be analytic. 
Being aware that flat contributions do not cause a perturbation series to have a vanishing radius of convergence, we 
still consider them an interesting feature of our model. 

Note that from the viewpoint of perturbation theory, we restrict ourselves to the one-loop 
approximation: the one-loop vacuum polarization in QED has no internal photon propagator, and hence does not iterate itself. 
Still, the non-linearity of Eq.(\ref{ODE1}) demands a nontrivial flat contribution. 

We investigate in Section \ref{secLanAvoi} how this contribution alters the behaviour of the running coupling and may in other cases
very well affect its growth so as to avert a Landau pole. However, Section \ref{secToyLP} explains how the flat component cannot 
prevent but at least shift the Landau pole of our toy model towards a higher momentum scale.

If we choose a perturbative approximation 
of $P(\alpha)$ of higher polynomial degree, the conclusions of this paper remain valid as long as we can arrange for flat 
contributions which alter the large $\alpha$ behaviour suitably, see \cite{BKUY09} and below. Note that the perturbation series 
of $P(\alpha)$ is defined through the primitive elements of the Hopf algebra underlying the quantum field theory under discussion, 
and is a variant of the contribution from the skeleton diagrams, see \cite{Y11,Y13} for details. 

We present the resulting toy model photon self-energy in Section \ref{secSelf} and study the impact the flat contribution has on the Green's
function. Although we \emph{can} relate the perturbative series of the function $P(\alpha)$ to a skeleton diagram expansion, 
we \emph{cannot} find a canonical diagrammatic interpretation of our toy photon self-energy in terms of a \emph{resummation scheme} like 
in the case  of renormalon chain \cite{FaSi97,Ben99} or rainbow approximations \cite{DelKaTh96,DelEM97,KY06}. 
Our approach is of a fundamentally different nature: we take a non-perturbative equation, solved it perturbatively for the first loop 
order and yet get an instantonic contribution.   

\section{Renormalization group recursion and Landau pole criterion}\label{secRG}

The Callan-Symanzik equation imposes a \emph{recursion identity} on the log-coefficient functions of the Green's function in the 
form (\ref{Green}) and has the ODE (\ref{ODE1}) as an immediate consequence. We quickly rederive this result by \cite{Y11} and 
subsequently investigate what can be said about the possible existence of a Landau pole given the behaviour of the function $P(\alpha)$.
We take a slightly different view from that in \cite{BKUY09} where  
\begin{equation}\label{YLanCrit}
 \mathcal{L}(P)=\int_{x_{0}}^{\infty} \frac{2dx}{x(\sqrt{1+4P(x)}-1)}  < \infty \hs{2} (x_{0}>0)
\end{equation}
is found to be a necessary and sufficient criterion for the existence of a Landau pole. 

\subsection*{Recursion identity} 
Our starting point is the Callan-Symanzik equation for the Green's function 
\begin{equation}\label{Call}
 \left(-\frac{\partial}{\partial L} + \beta(\alpha) \frac{\partial}{\partial \alpha} - \gamma(\alpha) \right)G(\alpha,L)=0
\end{equation}
into which we insert the Green's function in the form (\ref{Green}) and find
\begin{Proposition}
 The log-coefficient functions $\gamma_{k}(\alpha)$ of the Green's function $ G(\alpha,L)=1 - \gamma \cdot L$ with shorthand
$\gamma \cdot L := \sum_{k\geq 1}\gamma_{k}(\alpha)L^{k}$ 
satisfy the renormalization group (RG) recursion \cite{Y11}
\begin{equation}\label{rec2}
 (k+1)\gamma_{k+1}(\alpha) = \gamma(\alpha)(\alpha \partial_{\alpha}-1)\gamma_{k}(\alpha)
\end{equation}
\end{Proposition}
for all $k \geq 1$. 
\begin{proof}
If we plug $1 - \gamma \cdot L$ into the Callan-Symanzik RG equation (\ref{Call}), we find
\begin{equation}
 \begin{split}
0 &= \sum_{k\geq 1}k\gamma_{k}L^{k-1} - \beta \gamma' \cdot L -\gamma(1 - \gamma \cdot L) 
= \sum_{k\geq 2} k\gamma_{k}L^{k-1} - \beta \gamma' \cdot L +  \gamma(\gamma \cdot L)   \\
&= \sum_{k \geq 1} [(k+1)\gamma_{k+1} - \beta \gamma'_{k} + \gamma\gamma_{k}  ] L^{k}  
=  \sum_{k \geq 1} [ (k+1)\gamma_{k+1} -  \gamma(\alpha \gamma'_{k} - \gamma_{k}) ] L^{k} .  \qedhere
 \end{split}
\end{equation}
\end{proof}
If we now define $P(\alpha):=\gamma_{1}(\alpha)-2 \gamma_{2}(\alpha)$ and employ (\ref{rec2}) for $k=1$, we get the non-linear ODE
\begin{equation}\label{ODE}
 \gamma(\alpha) + \gamma(\alpha)(1-\alpha \partial_{\alpha})\gamma(\alpha) = P(\alpha)
\end{equation}
for the \emph{anomalous dimension} $\gamma(\alpha)=\gamma_{1}(\alpha)$ of the photon, where the definition of $P(\alpha)$ was originally inspired by a study 
of the photon's Dyson-Schwinger equation (see \cite{Y11}). It is this very context in which the perturbative series of $P(\alpha)$ can be
given a diagrammatic interpretation in terms of skeleton diagrams \cite{Y11,Y13}. 

\subsection*{Landau pole} So what use is this equation? Firstly, we shall now see that it implies a sufficient criterion for the existence of a Landau pole which is 
equivalent to a necessary condition for QED to have no such pole. 
Secondly, it suggests new possibilities on the non-perturbative frontline.
What do we know about $P(\alpha)$?  
First note that $\beta(\alpha)>0$, and hence $P(\alpha)>0$ for small $\alpha >0$ by how it is defined. This
function may have zeros: at a point $\alpha_{0} \in (0,\infty)$, where $P(\alpha_{0})=0$ we see that by
\begin{equation}\label{Palpha}
0=P(\alpha_{0})= \gamma(\alpha_{0})[1 + \gamma(\alpha_{0})-\alpha_{0} \gamma'(\alpha_{0})] 
\end{equation}  
we have either $\gamma(\alpha_{0})=0$ and thus $\beta(\alpha_{0})=0$ or
\begin{equation}\label{diff}
 1 + \gamma(\alpha_{0})-\alpha_{0} \gamma'(\alpha_{0})=0.
\end{equation}
We exclude the possibility that $P(\alpha)$ has an infinite number of zeros and compare the following two assumptions from a 
physical point of view:
\begin{itemize}
\item[(H1)] \emph{$P(\alpha)$ has no nontrivial zero, i.e. no zero other than $\alpha_{0}=0$}.
\item[(H2)] \emph{The anomalous dimension $\gamma(\alpha)$ has no nontrivial zero whereas $P(\alpha)$ does have a finite number of zeros}.
\end{itemize}
Notice that (H1) implies $\gamma(\alpha)>0$ for all $\alpha > 0$. The next two propositions will help us to decide which of these 
assumptions is stronger.  

\begin{Proposition}[Asymptotics of anomalous dimension I]\label{Prop:Asy1}
Assume $P(\alpha)$ vanishes nowhere except at the origin. Then there exist a constant $A>0$ and a point $\alpha^{*}\in \mathbb{R}_{+}$ 
such that 
\begin{equation}\label{asym1}
 \gamma(\alpha) < A \alpha - 1 
\end{equation}
for all $\alpha > \alpha^{*}$.  
\end{Proposition}
\begin{proof}
Pick any $\alpha^{*} \in (0, \infty)$ and set $A:=[1+\gamma(\alpha^{*})]/\alpha^{*}$. By (\ref{Palpha}), the assumptions imply 
\begin{equation}\label{deriv1}
\forall \alpha \in \mathbb{R}_{+} : \hs{0.1}  \gamma'(\alpha) < \frac{1+\gamma(\alpha)}{\alpha} \hs{0.2}  
\Rightarrow \hs{0.2} \gamma'(\alpha^{*}) < \frac{1+\gamma(\alpha^{*})}{\alpha^{*}} = A  
\end{equation}
and thus $h_{A}(\alpha):= A \alpha -1$ is the line that meets $\gamma(\alpha)$ but has stronger growth at the point
$\alpha=\alpha^{*}$. 
Hence $h(\alpha) := \gamma(\alpha)-h_{A}(\alpha)$ satisfies $h(\alpha^{*})=0$ and $h'(\alpha^{*})<0$. Consequently, there is an 
interval $I_{\varepsilon}(\alpha^{*}):=(\alpha^{*},\alpha^{*}+\varepsilon)$ such that $h(\alpha)<0$ for all 
$\alpha \in I_{\varepsilon}(\alpha^{*})$. For a sign change of $h(\alpha)$
there must be a point $\overline{\alpha}>\alpha^{*}$ where $h(\overline{\alpha})<0$ and $h'(\overline{\alpha})=0$ which implies the 
contradiction
\begin{equation}
 0= h'(\overline{\alpha}) = \gamma'(\overline{\alpha}) - A < \frac{1+\gamma(\overline{\alpha})}{\overline{\alpha}} 
- A = \frac{h(\overline{\alpha})}{\overline{\alpha}}. 
\end{equation}
\end{proof}
Note that the asymptotics of (\ref{asym1}) does not touch on the Landau pole question of QED: the growth of the $\beta$-function may or may not
be strong enough for a Landau pole to exist. Regarding the second assumption (H2), we will see that we need the extra property that 
$P(\alpha)>0$ for large enough $\alpha \in \mathbb{R}_{+}$ to not have a Landau pole enforced.  

\begin{Proposition}[Asymptotics of anomalous dimension II]\label{Prop:Asy2}
Assume $\gamma(\alpha)$ vanishes nowhere other than at the origin and $P(\alpha)$ has a finite number of zeros such that it becomes 
negative for sufficiently large $\alpha$, i.e. there is an $\alpha^{*} \in \mathbb{R}_{+}$ with
$P(\alpha)<0$ for all $\alpha > \alpha^{*}$. Then there exists a constant $A>0$ such that 
\begin{equation}\label{asym2}
 \gamma(\alpha) > A \alpha - 1 
\end{equation}
for all $\alpha > \alpha^{*}$ and QED has a Landau pole. 
\end{Proposition}
\begin{proof}
With $A$ defined as above, the assumptions imply
\begin{equation}\label{deriv2}
 \gamma'(\alpha^{*}) > \frac{1+\gamma(\alpha^{*})}{\alpha^{*}} = A
\end{equation}
and hence this time $h_{A}(\alpha)= A \alpha -1$ is the line that meets $\gamma(\alpha)$ at the point $\alpha=\alpha^{*}$ but 
is growing less there. 
Consequently, $h(\alpha) = \gamma(\alpha)-h_{A}(\alpha)$ satisfies $h(\alpha^{*})=0$ and $h'(\alpha^{*})>0$. For a sign change of $h(\alpha)$
there must be a point $\overline{\alpha}>\alpha^{*}$ where $h(\overline{\alpha})>0$ and $h'(\overline{\alpha})=0$ which implies the 
contradiction
\begin{equation}
 0= h'(\overline{\alpha}) = \gamma(\overline{\alpha}) - A > \frac{1+\gamma(\overline{\alpha})}{\overline{\alpha}} 
- A = \frac{h(\overline{\alpha})}{\overline{\alpha}}. 
\end{equation}
The second assertion holds because of
\begin{equation}
|\int_{\alpha^{*}}^{\infty} \frac{dx}{\beta(x)} | \leq |\int_{\alpha^{*}}^{\infty} \frac{dx}{x h_{A}(x)} | < \infty,
\end{equation}
where the first integral is discussed in Section \ref{secLanAvoi}.  
\end{proof}
In other words, on the assumption (H2), QED can only be free of a Landau pole if $P(\alpha)\geq 0$ for large enough $\alpha$.
One may therefore on purely physical grounds view (H1) to be less strong than (H2). 
With the latter Proposition and the definition of $P(\alpha)$, we arrive at

\begin{Corollary}[Landau pole criterion]
QED has a Landau pole if $\gamma_{1}(\alpha)<2\gamma_{2}(\alpha)$ for large enough $\alpha > 0$. A necessary condition for the 
non-existence of a Landau pole therefore is given by 
\begin{equation}\label{cond}
 \gamma_{1}(\alpha) \geq 2 \gamma_{2}(\alpha) \hs{1} \mbox{for sufficiently large $\alpha$} ,
\end{equation}
i.e. the second log-coefficient function must not win out over the first. 
\end{Corollary}

We shall see in Section \ref{secToyLP} that though our toy model adheres to the necessary condition (\ref{cond}), it does have a Landau pole.

\subsection*{Perturbative Expansion} To study the ODE in (\ref{ODE}) perturbatively, we expand $P(\alpha)$ in $\alpha$.
Let us see what the first coefficients are. Diagrammatically, the Green's function reads in terms of Feynman diagrams,  
\begin{equation}\label{GreenDiag}
\begin{split}
\tree{0.1}{-.2}{0.2}{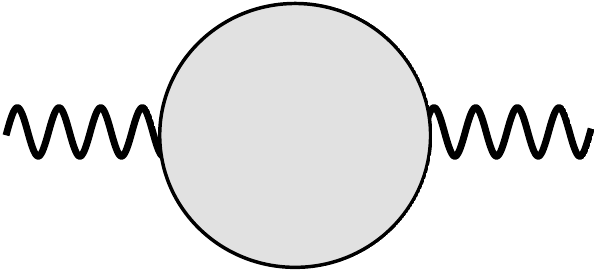}{0.1} = 1 - \tree{0.1}{-.2}{0.2}{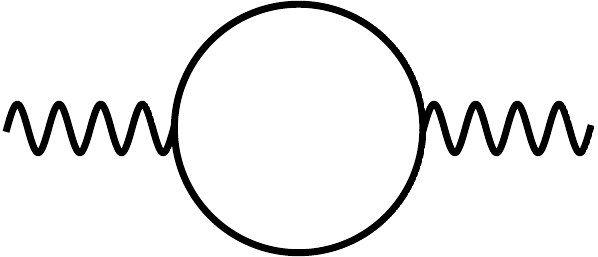}{0.1} \alpha - (\tree{0.1}{-.2}{0.2}{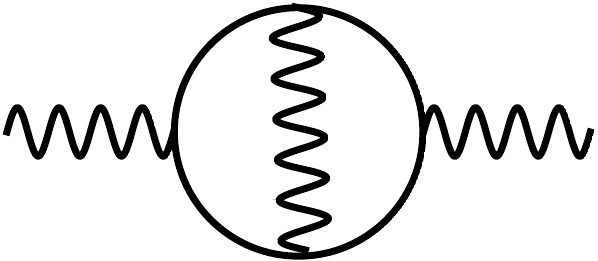}{0.1} +\tree{0.1}{-.2}{0.2}{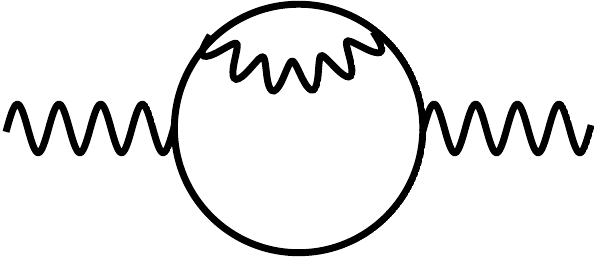}{0.1} 
+ \tree{0.1}{-.2}{0.2}{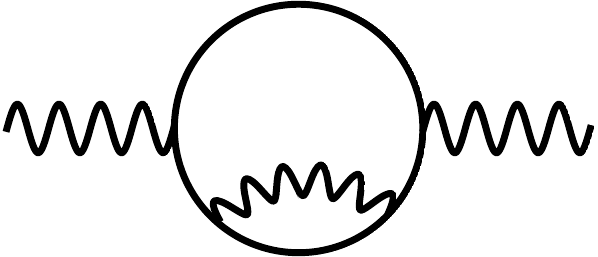}{0.1}) \alpha^{2} - ... ,
\end{split}
\end{equation}
which, as an expansion, can be seen as a formal power series
\begin{equation}\label{formSeries}
 X(\alpha)=1-\sum_{k\geq 1}c_{k}\alpha^{k},
\end{equation}
with linear combinations $c_{k}$ of Feynman graphs as coefficients, i.e. in particular 
\begin{equation}
c_{1} = \tree{0.1}{-.2}{0.2}{1qedph1}{0.1} , \hs{2} c_{2} = \tree{0.1}{-.2}{0.2}{1qedph}{0.1} +\tree{0.1}{-.2}{0.2}{1qedpho}{0.1} 
+ \tree{0.1}{-.2}{0.2}{1qedphu}{0.1}
\end{equation}
and so on\footnote{The first term '$1$' corresponds to the bare propagator.}. Let now $\Ho$ be the free vector space spanned by all 1PI 
divergent photon propagator diagrams. Then the renormalized Feynman rules are 
characterized by linear maps $\phi_{L}:\Ho \rightarrow \mathbb{C}[L]$ that  evaluate a graph $\Gamma$ to a polynomial in the external 
momentum parameter $L$ of the form   
\begin{equation}
\phi_{L}(\Gamma) = \sum_{j=1}^{n} \sigma_{j}(\Gamma) L^{j},
\end{equation}
where $\sigma_{j}(\Gamma) \in \mathbb{C}$ is the coefficient of the $j$-th power of $L=\ln(-q^{2}/\mu^{2})$ and $n$ the degree of the
polynomial which is bounded above by the loop number of $\Gamma$. This implicitly defines linear maps $\sigma_{j}:\Ho \rightarrow \mathbb{C}$ with the property that $\sigma_{j}(\Gamma)=0$ 
for all $j > N$, if $\Gamma$ has $N$ loops. With these properties, the evaluation map $\phi_{L}$ can now be decomposed into
\begin{equation}\label{Decomp}
 \phi_{L} = \sum_{j=0}^{\infty}L^{j}\sigma_{j} = \sigma_{0} + L \sigma_{1} + L^{2} \sigma_{2} + ...
\end{equation}
where $\sigma_{0}(1)=1$ and $\sigma_{0}(\Gamma)=0$ for any Feynman graph. Note that for $j\geq 1$, one has $\sigma_{j}(1)=0$, i.e. the
bare propagator evaluates to the trivial polynomial $\phi_{L}(1)=1$. This is because we are dealing with Feynman rules for a form factor
here. If we apply the decomposition (\ref{Decomp}) to the formal series (\ref{formSeries}), we get the Green's function 
\begin{equation}
G(\alpha,L) = \phi_{L}(X(\alpha))= \sum_{j\geq 0}\sigma_{j}(X(\alpha))L^{j} = 1 + \sum_{j\geq 1}\sigma_{j}(X(\alpha))L^{j}
\end{equation}
in which we identify $\sigma_{j}(X(\alpha))=-\gamma_{j}(\alpha)$ for $j\geq 1$, i.e. by linearity of $\sigma_{j}$ we have the asymptotic
series
\begin{equation}
 \gamma_{j}(\alpha) = \sum_{k\geq 1} \sigma_{j}(c_{k})\alpha^{k} 
= \sum_{k\geq j} \sigma_{j}(c_{k})\alpha^{k}  =: \sum_{k \geq j} \gamma_{j,k} \alpha^{k},
\end{equation}
where $\gamma_{j,k}=\sigma_{j}(c_{k})$ is the $k$-th (asymptotic) Taylor coefficient of $\gamma_{j}(\alpha)$. 
Note that $c_{k} \in \Ho$ is in general a linear combination of $k$-loop vacuum polarization
graphs and hence $\sigma_{j}(c_{k})=0$ for $j > k$. 
This entails that the asymptotic expansion of $\gamma_{j}(\alpha)$ starts with the $j$-th coefficient $\gamma_{j,j}$ which is also
implied by the recursion in (\ref{rec2}). However, with these maps 
at hand, we see that the perturbative expansion of the function $P(\alpha)$ is given by
\begin{equation}
 P(\alpha)=\sum_{k\geq 1}[\sigma_{1}(c_{k})-2\sigma_{2}(c_{k})]\alpha^{k}. 
\end{equation}
Note that $\phi_{L}(c_{1})=\sigma_{1}(c_{1})L$ and $\phi_{L}(c_{2})=\sigma_{1}(c_{2})L + \sigma_{2}(c_{2})L^{2}$. From the results 
\begin{equation}
\phi_{L}(c_{1}) = \phi_{L}(\tree{0.1}{-.2}{0.2}{1qedph1}{0.1}) = \frac{L}{3\pi}, \hs{0.5} 
\phi_{L}(c_{2}) = \phi_{L}(\tree{0.1}{-.2}{0.2}{1qedph}{0.1} +\tree{0.1}{-.2}{0.2}{1qedpho}{0.1} 
+ \tree{0.1}{-.2}{0.2}{1qedphu}{0.1}) = \frac{L^{2}}{4\pi^{2}},
\end{equation}
found in \cite{GoKaLaSu91}, we read off the coefficients of $P(\alpha)$ to find
\begin{equation}\label{P}
 P(\alpha)= \frac{1}{3} \frac{\alpha}{\pi} - \frac{1}{2} \left(\frac{\alpha}{\pi}\right)^{2} + \mathcal{O}(\alpha^{3}) .
\end{equation}
Given the perturbative series of $P(\alpha)$, the ODE in (\ref{ODE}) determines the anomalous dimension $\gamma(\alpha)$ perturbatively:
let $u_{j}:=\gamma_{1,j}$ be the asymptotic coefficients of $\gamma(\alpha)$ and $r_{1}, r_{2}, ...$ those of $P(\alpha)$, 
then (\ref{ODE}) imposes
\begin{equation}\label{AsymRecursion}
 r_{k} = u_{k} + \sum_{l=2}^{k-1} (1-l) u_{l} u_{k-l}
\end{equation}
giving a nice recursion \cite{Y13}, 
\begin{equation}
 r_{1} = u_{1}\  , \hs{1} r_{2} = u_{2} \ , \hs{1} r_{3} = u_{3} - u_{2} u_{1} \ , \hs{0.2} ... \hs{0.2}
\mbox{and so on}.
\end{equation}
Though we know next to nothing about $P(\alpha)$, we assume it to be \emph{non-analytic} with a non-convergent asymptotic series 
which is Gevrey-1, that is, its Borel sum
\begin{equation}
B[P](\alpha) = \sum_{k\geq 1}\frac{r_{k}}{k!}\alpha^{k}
\end{equation}
has a non-vanishing radius of convergence. We furthermore expect the Borel sum to have an analytic continuation 
$\widetilde{B}[P](\alpha)$ such that $P(\alpha)$ equals the Borel-Laplace transform
\begin{equation}
P(\alpha) = \int_{0}^{\infty} dt \ e^{-t} \widetilde{B}[P](\alpha t)  
\end{equation}
up to flat contributions. Clearly, the recursion (\ref{AsymRecursion}) is blind to such contributions. Perhaps not surprisingly, 
\cite{BKUY10} have found an upper bound for the difference of two different solutions of (\ref{ODE}) in terms of a flat function:

\begin{Proposition}\label{KYB}
Let $P \in \mathscr{C}^{2}(\mathbb{R}_{+})$ be positive with $P(0^{+})=0$ and $P'(0^{+}) \neq 0$. Then two solutions $\gamma$ and
$\widetilde{\gamma}$ of the ODE (\ref{ODE}) differ by a flat function, more precisely,
\begin{equation}
 |\gamma(\alpha)-\widetilde{\gamma}(\alpha)| \leq E \alpha \exp(-F/\alpha),   \hs{2} \forall \alpha \in [0,\alpha_{0}]
\end{equation}
where the the constants $E,F>0$ depend on $\alpha_{0}>0$.
\end{Proposition}
\begin{proof}
 See Theorem 5.1 in \cite{BKUY10} .
\end{proof}

To round off this section, we mention for the sake of completeness that a more general version of (\ref{ODE}) pertaining not just 
to QED has been studied in \cite{BKUY09} with the following result. On the assumption that $P(\alpha)$ is twice differentiable and 
strictly positive on $(0,\infty)$ with $P(0)=0$, a global solution exists iff $J(P):= \int_{x_{0}}^{\infty}dz P(z) z^{-3}$ converges for 
some $x_{0}>0$. If furthermore $P(\alpha)$ is everywhere increasing, then 
there exists a 'separatrix': a global solution that separates all global solutions from those existing only on a finite interval. 
We shall see in the next section how this latter situation arises in our approximation $P(\alpha)=c\alpha$ and ensures that the 
family of solutions (\ref{LamberW}) covers the whole set of solutions. Moreover, the separatrix picks out the very physical
solution that corresponds to a $\beta$-function whose growth is weakest among those of all other possible physical cases. Although it 
is not weak enough to avert a Landau pole, it may very well be in the case of the 'true' $P(\alpha)$. Because $P(\alpha)$ dictates the
behaviour of the separatrix, it strikes us to be an interesting object to scrutinize and maybe find an ODE for. The avoidance of a
Landau pole may then turn out to be simply encoded in the form of a boundary condition for $P(\alpha)$.

\section{First order non-analytic approximation}\label{secFoApp}

We set $c:=1/(3\pi)$ and choose $P(\alpha)=c \alpha$ for a first order approximation. Now, the ODE in (\ref{ODE}) takes the form
\begin{equation}\label{ode}
 \gamma(\alpha) + \gamma(\alpha)(1-\alpha \partial_{\alpha})\gamma(\alpha) = c \alpha. 
\end{equation}
This equation has already been studied in \cite{BKUY09} where the reader can find a direction field for the anomalous dimension 
$\gamma(\alpha)$. We shall review their results and expound them in somewhat more detail. 
It is an easy exercise to prove that 
\begin{equation}\label{anoDim}
\gamma(\alpha) = c\alpha[1 + W(\xi e^{-\frac{1}{c \alpha}})]
\end{equation}
provides a family of solutions with parameter $\xi:=(\gamma(1/c)-1)e^{\gamma(1/c)}$ which is fixed by the initial condition
for $\gamma(\alpha)$ at $\alpha=1/c$. This follows from
\begin{equation}
 \gamma(1/c)-1 = W(\xi e^{-1}) = \xi e^{-1} e^{-W(\xi e^{-1})} = \xi e^{-[1+W(\xi e^{-1})]} = \xi e^{-\gamma(1/c)},
\end{equation}
where we have used that the \emph{Lambert W function} $W(x)$ is defined as the inverse function of $x \mapsto x \exp(x)$ and 
therefore characterized by
\begin{equation}\label{Lampop}
 x = W(x)e^{W(x)}.
\end{equation}
This function has two branches, denoted by $W_{0}$ and $W_{-1}$ (see Figure \ref{LamFig}) and emerges in physics whenever the identity 
(\ref{Lampop}) may be exploited to solve a transcendental equation\footnote{See for example the QCD-related papers \cite{GaKaG98} 
and \cite{Nest03}.}.
\begin{figure}[ht] 
\includegraphics[height=7cm]{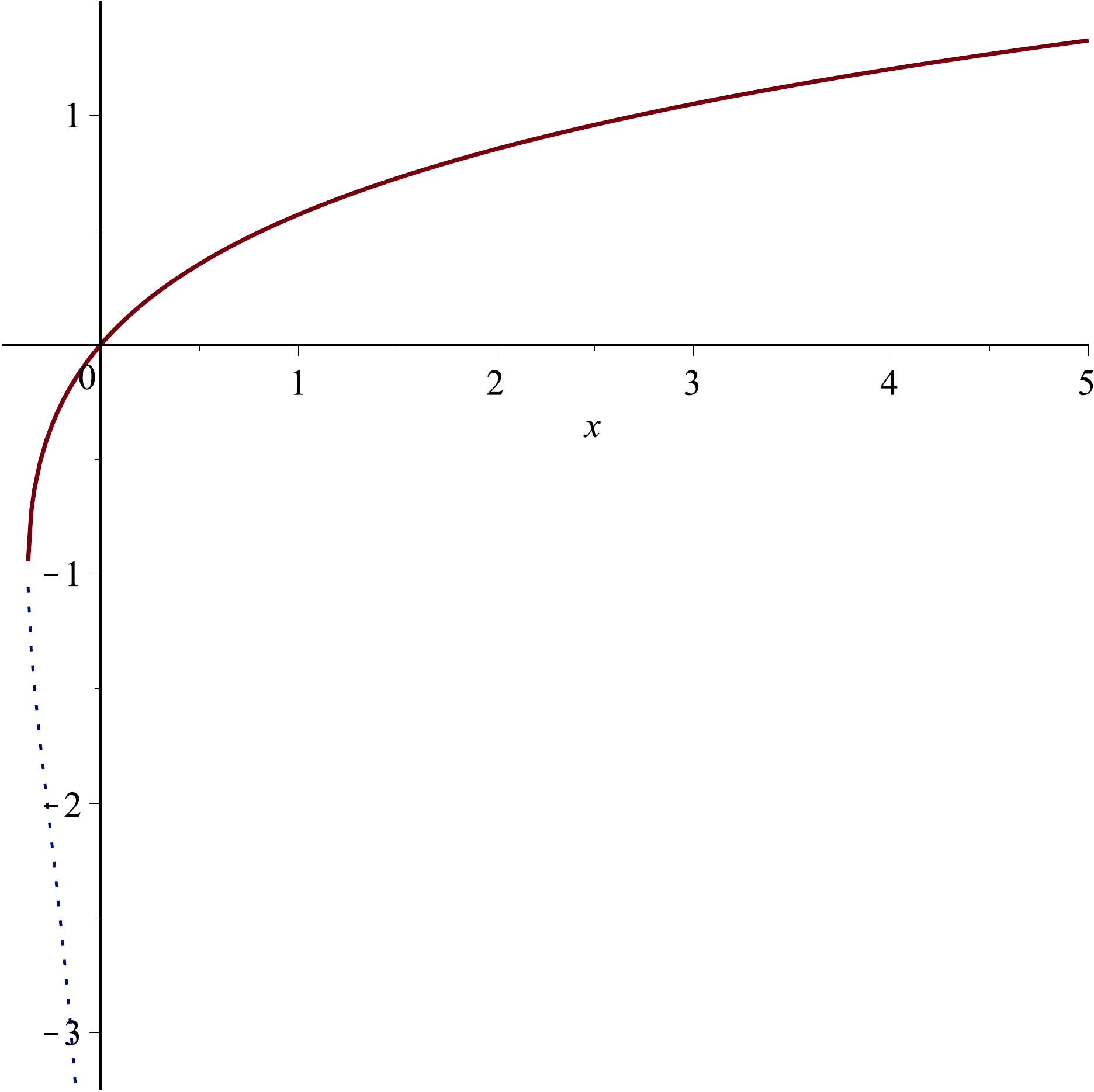}   \graph{-2.5}{1.5}{0.2}{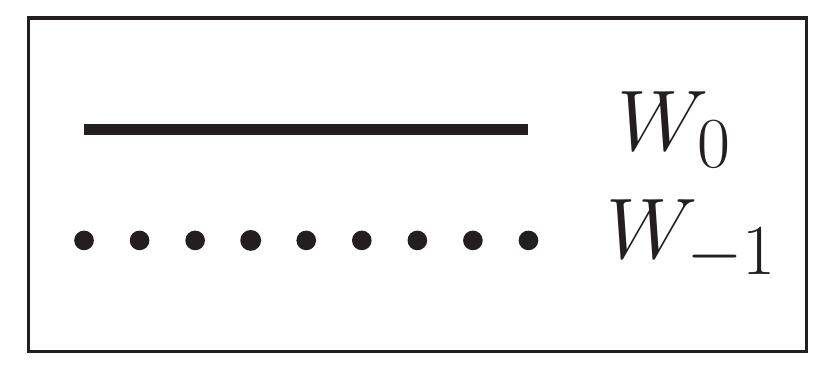}{0}
\caption{ \small The two branches of the Lambert W function. Note that the second branch $W_{-1}$ (dotted line) is restricted to the 
interval $[-1/e,0)$ and vanishes nowhere.} \label{LamFig} \end{figure}
We shall ignore the second branch $W_{-1}(x)$, for the following reason: it is only defined on the half-open interval $[-1/e,0)$ and 
coerces us to choose $\xi<0$. On this interval, it rapidly drops down an abyss where one finds $W_{-1}(0^{-})=-\infty$. But although
it turns out that $\gamma(0^{+})=0$ in this branch, one finds $\gamma(\alpha)< 0$ for all couplings which entails 
$\beta(\alpha)=\alpha \gamma(\alpha) < 0$ for the $\beta$-function.
As this is not what we would call QED-like behaviour, we discard this branch. In contrast to this, we will see that the first 
branch $W_{0}(x)$ serves our purposes perfectly well. We will denote it by $W(x)$.  

Because our approximation does only hold for very small values of the coupling parameter $\alpha$
and, for example
\begin{equation}
 \frac{c}{137} W(\xi e^{-\frac{137}{c}}) \sim 10^{-562}
\end{equation}
with $\xi=-e^{-1}$, we shall scale away $c$ such that $c \alpha \rightarrow \alpha$ without renaming of functions and view all of the 
following results arising from $P(\alpha)=\alpha$ as those of an interesting toy model.  

\subsection*{$\beta$-function} The point $\xi_{*}:=-e^{-1}$ turns out to be critical for the $\beta$-function
\begin{equation}
 \beta(\alpha)= \alpha \gamma(\alpha) =  \alpha^{2}[1+W(\xi e^{-\frac{1}{\alpha}})] .
\end{equation}
This function has a nontrivial zero if we choose $\xi < \xi_{*}$ and no such zero otherwise: 
the only way the $\beta$-function can vanish at some point $\alpha_{0} \in (0,\infty)$ is when 
\begin{equation}
W(\xi e^{-\frac{1}{\alpha_{0}}})=-1 
\end{equation}
which by $x=W(x) e^{W(x)} = -1 e^{-1}$ implies 
\begin{equation}
 \xi e^{-\frac{1}{\alpha_{0}}} = -e^{-1}  \hs{1} \Rightarrow \hs{1} \alpha_{0} = \frac{1}{1+\ln |\xi|}.
\end{equation}
This point escapes into infinity for $\xi=\xi_{*}$ and reappears on the negative side of the real axis for $\xi > \xi_{*}$ 
where it is no longer of interest as a zero of the $\beta$-function. The choice $\xi=\xi_{*}$ corresponds to the initial condition 
\begin{equation}
 \gamma(1) = 1+W(-e^{-2})  \approx 0.841
\end{equation}
and characterizes the \emph{separatrix} $\beta$-function $\beta^{*}(\alpha)= \alpha \gamma^{*}_{1}(\alpha)$. Further inspection 
shows that all choices $\xi < \xi_{*}$
with a nontrivial zero $\alpha_{0}>0$ are unphysical: their solution $\beta(\alpha)$ simply ceases to exist at $\alpha_{0}$
and has a divergent derivative at this point, i.e. $\beta(\alpha_{0})=0$ and $\beta'(\alpha_{0})=-\infty$. We therefore conclude 
that only $\xi \geq \xi_{*}=-e^{-1}$ is physically permissible for further consideration. Note that 
the usual one-loop approximation for the $\beta$-function corresponds to the case $\xi=0$, which is also physical. 
Figure \ref{Beta} shows examples for different choices of $\xi$. 

\begin{figure}[ht] 
\includegraphics[height=7cm]{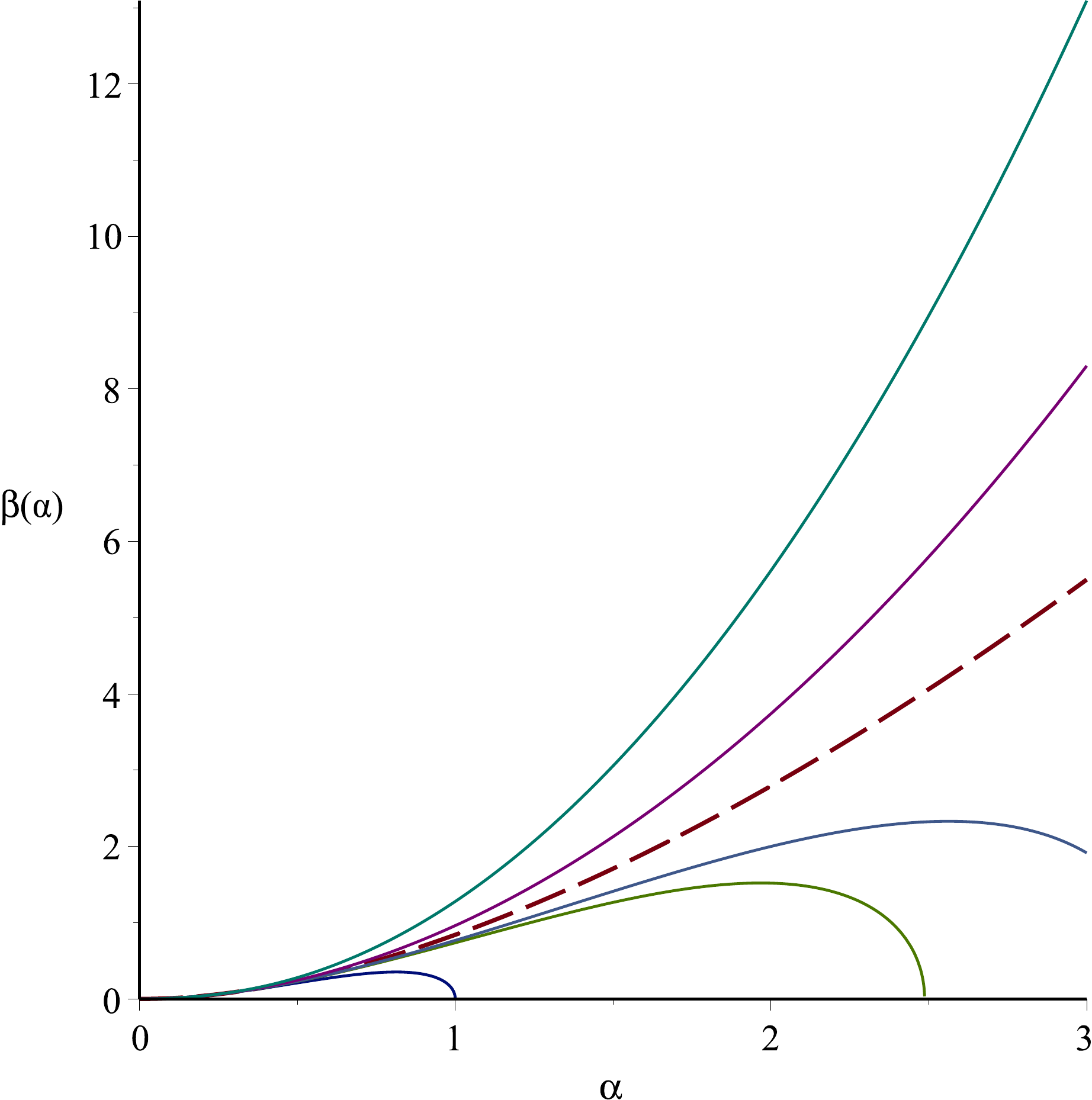}  
\caption{ \small The $\beta$-function for different choices of $\xi$: only the 
separatrix corresponding to $\xi=\xi_{*}$ (dashed line) and the curves above it with $\xi > \xi_{*}$ are physical, whereas those with 
a zero for $\xi< \xi_{*}$ are not.} 
\label{Beta} 
\end{figure}

The possible solutions reflect the results of \cite{Y11} alluded to in the previous section: the separatrix $\beta^{*}(\alpha)$ 
separates global solutions from those with a finite interval of definition. As a result, the above family of solutions in (\ref{anoDim})
covers the set of all possible solutions.  

\section{Flat contributions}\label{secFlatCon}

Let for the sake of the next assertions $D:=1-\alpha \partial_{\alpha}$ and 
\begin{equation}
 \F:= \left\{ \ f \in \mathscr{C}^{\infty}(\mathbb{R}_{+}) \ \vert  
\ \forall n \in  \mathbb{N}: \  \lim_{\alpha \downarrow 0} \partial_{\alpha}^{n} f(\alpha) = 0 \ \right\}
\end{equation}
be the algebra of all flat functions. To study flat contributions, we have to make an assumption 
about $P(\alpha)$ against our better judgement due to a mathematical subtlety. We still believe our results to be of value. The 
issue is this: because the algebra $\F$ is a subspace in the space of smooth functions $\mathscr{C}^{\infty}(\mathbb{R}_{+})$, there
exists a projector $\pi_{\F}:\mathscr{C}^{\infty}(\mathbb{R}_{+}) \rightarrow \F$ such that $f \in \mathscr{C}^{\infty}(\mathbb{R}_{+})$ can be 
decomposed into
\begin{equation}\label{decomp}
 f = (\id - \pi_{\F})(f) + \pi_{\F} (f) = f_{0} + f_{1},
\end{equation}
where $f_{1}:= \pi_{\F}(f) \in \F$ is flat. However, there is surely not just one projector and hence not just one 
possible decomposition of $f$ into 'flat' and 'non-flat': take any flat $g \in \F$, then 
\begin{equation}
 f= f_{0} - g + f_{1} + g = \widetilde{f}_{0} + \widetilde{f}_{1}
\end{equation}
with $\widetilde{f}_{0}:=f_{0} - g$ being the non-flat and $\widetilde{f}_{1} = f_{1} + g$ the flat part. As a consequence, there is 
no unique decomposition of the desired kind and things get cloudy at the attempt to find a strict mathematical definition of 
'flat contribution'.  However, this is not so if we restrict ourselves
to the subspace of functions $f \in \mathscr{C}^{\infty}(\mathbb{R}_{+})$ with convergent Taylor series 
\begin{equation}
 (Tf)(\alpha):=\sum_{k\geq 0}\frac{1}{k!}f^{(k)}(0^{+})\alpha^{k}
\end{equation}
at zero having an analytic continuation $\widetilde{T}f$ to the full half-line $[0,\infty)$. An easy example is 
\begin{equation}
 f(\alpha)=\frac{1}{1+\alpha} + e^{-1/\alpha}.
\end{equation}
Its Taylor series $\sum_{j\geq 0}(-1)^{j}\alpha^{j}$ is convergent, enjoys an analytic continuation to $[0,\infty)$ 
and yet it converges nowhere to $f(\alpha)$. We denote the algebra of these functions by $\A$
and define the projector $\pi_{\F}:\A \rightarrow \F$ as the (uniquely determined) linear operator that 
subtracts the analytic continuation of the convergent Taylor series, i.e.
\begin{equation}
 \pi_{\F}(f):=f-\widetilde{T}f = (\id -\widetilde{T})f
\end{equation}
is in $\F$ and the decomposition $f=\widetilde{T}f + \pi_{\F}(f)$ is unique. We therefore have the decomposition
\begin{equation}\label{Decompo}
 \A = \A_{0} \oplus \F
\end{equation}
with $\A_{0}:=(\id - \pi_{\F})\A$ being also an algebra. We write $\pi_{\A_{0}} := (\id - \pi_{\F})=\widetilde{T}$ for its projector. 
In fact, $\A_{0}$ is the well known algebra of analytic functions on $[0,\infty)$. It is invariant under differential operators 
and hence a \emph{differential algebra}. In particular, this implies $D$-invariance, i.e. $D \A_{0} \subset \A_{0}$. The flat algebra
also has this property. For later reference, we list its properties:
\begin{itemize}
 \item[(i)] $\F$ is $D$-invariant, that is, $D \F \subseteq \F$. 
 \item[(ii)] The product of any function in $\A$ and a flat function is flat: $\A \F \subset \F$, i.e. $\F$ is an ideal in the 
algebra $\A$.
\end{itemize}
In summary, $\A$ is the class of functions $f \in \mathscr{C}^{\infty}(\mathbb{R}_{+})$ with convergent Taylor 
series(at $\alpha=0$) \emph{that do not converge to} $f$ only if $\pi_{\F}(f) \neq 0$, i.e. if $f$ features a nontrivial flat part(which 
renders it nonanalytic).  
Note that the operator $D$ has the one-dimensional kernel $\ker D = \mathbb{R} \alpha \subset \A_{0}$ and we therefore have a 
third property:
\begin{itemize}
 \item [(iii)] If $f \in \mathbb{R} \alpha + \F$, then $Df \in \F$.
\end{itemize}
We shall draw on (i)-(iii) in the proofs of the following assertions which we view as interesting on the following grounds. 

Being aware that $\gamma(\alpha)$ and almost surely also $P(\alpha)$ are nonanalytic functions with divergent Taylor series,
we would like to point out that within our approach of approximating $P(\alpha)$ perturbatively by a polynomial in $\alpha$ and hence by a function 
in the class $\A_{0}$, it makes perfect sense to us to assume that $\gamma(\alpha)$ is at most in the class $\A$: 
the coefficient recursion in (\ref{AsymRecursion}) can only be expected to lead to a divergent series of $\gamma(\alpha)$ in 
mathematically contrived situations. By allowing $P(\alpha)$ to be in $\A$, we go one doable step \emph{beyond} perturbation theory. 
'Doable' because the decomposition (\ref{Decompo}) is mathematically well-defined in a way that enables us to get a grip on the 
otherwise vague concept of flat contributions which the $\beta$-function may or may not feature.

\begin{Claim}[Flat perturbations I]
Let $P \in \A$ with a nontrivial flat part: $\pi_{\F}(P) \neq 0$. Then any solution of the ODE
\begin{equation}\label{ode1}
 \gamma(\alpha) + \gamma(\alpha)D\gamma(\alpha) = P(\alpha)
\end{equation}
does also have a nontrivial flat part, that is, $\pi_{\F}(\gamma) \neq 0$. 
\end{Claim}
\begin{proof} Let $\pi_{\F}(\gamma) = 0$. Then follows $\pi_{\F}(D\gamma) =0$ and also $\pi_{\F}(\gamma D \gamma)=0$ by $\A_{0}$ being
a differential algebra. This entails $\pi_{\F}(P)=0$.  
\end{proof}

Clearly, the reason why a flat contribution may pop up in the anomalous dimension $\gamma$ even when $\pi_{\F}(P)=0$ is that (\ref{ode1})
implies
\begin{equation}
 \pi_{\F}(\gamma) + \pi_{\F}(\gamma D \gamma) = 0
\end{equation}
which can be massaged into a differential equation for the flat function $\pi_{\F}(\gamma)$ and has solutions beyond the trivial one.
The next assertion specializes in the toy case $P(\alpha) \in \alpha + \F$ and reveals how the anomalous dimension $\gamma$ is affected.

\begin{Claim}[Flat perturbations II]\label{FLATP}
Let $\gamma(\alpha)$ be a solution of 
\begin{equation}\label{noflatPert}
 \gamma(\alpha) + \gamma(\alpha)D \gamma(\alpha) = \alpha .
\end{equation}
and $\overline{\gamma}(\alpha)$ of its flatly perturbed version 
\begin{equation}\label{flatPert}
 \overline{\gamma}(\alpha) + \overline{\gamma}(\alpha)D\overline{\gamma}(\alpha)=\alpha + f(\alpha),
\end{equation}
where $f \in \F$. Then $\overline{\gamma}-\gamma \in \F$, i.e. a flat perturbation of the rhs of 
(\ref{noflatPert}) leads to a flat perturbation of its solution. 
\end{Claim}
\begin{proof}
Recall from Section \ref{secFoApp} that any solution of (\ref{noflatPert}) satisfies
\begin{equation}
 \gamma(\alpha) \in \alpha + \F
\end{equation}
and thus $\pi_{\A_{0}}(\gamma)=\alpha$. Because $P(\alpha)$ determines the perturbation series of $\gamma(\alpha)$ uniquely
through (\ref{AsymRecursion}) where flat parts do not participate, the perturbation series of $\overline{\gamma}$ and $\gamma$ coincide. 
Trivially, this means that if $\gamma$ has a convergent Taylor series, so does $\overline{\gamma}$. Hence 
$\gamma, \overline{\gamma} \in \A$ and there is a decomposition $\overline{\gamma}-\gamma=h_{0}+h_{1}$, 
where $h_{0}=\pi_{\A_{0}}(\overline{\gamma}-\gamma)$ and $h_{1} \in \F$ is flat. We will show that the function $h_{0}$ vanishes 
everywhere. Subtracting (\ref{noflatPert}) from (\ref{flatPert}) yields
\begin{equation}\label{subtra}
\overline{\gamma}-\gamma + (\overline{\gamma}-\gamma) D \overline{\gamma} + \gamma D(\overline{\gamma}-\gamma) = f.
\end{equation}
To get rid of all the flat stuff, we apply the projector $\pi_{\A_{0}}=(\id-\pi_{\F})$ to both sides of this equation and obtain
\begin{equation}\label{proj}
 h_{0} + h_{0}Dh_{0} + \alpha D h_{0} = 0,
\end{equation}
where we have used $D(\gamma + h_{1}) \in \F$ and (i)-(iii). We can rewrite (\ref{proj}) in the form
\begin{equation}
 (h_{0} + \alpha) + (h_{0} + \alpha)D (h_{0} + \alpha) = \alpha,
\end{equation}
and find that $\gamma = h_{0} + \alpha$. This means $h_{0} \in \F$ and therefore $\pi_{\F}(h_{0})=h_{0}$. Consequently, 
on account of $\pi_{\F}(h_{0})=0$, we see that $h_{0}=0$.   
\end{proof}

The next Proposition generalizes this latter assertion. 

\begin{Proposition}[Flat Perturbations III]
 Let $P \in \A_{0}$ be such that $P(0^{+})=0$, $P'(0^{+})\neq 0$ and $P(\alpha)>0$ for $\alpha >0$. Then any solution 
$\overline{\gamma}$ of 
\begin{equation}\label{flaPertis}
 \overline{\gamma}(\alpha) + \overline{\gamma}(\alpha) D \overline{\gamma}(\alpha) = P(\alpha) + f(\alpha)
\end{equation}
with $f \in \F$ differs from a solution of $\gamma(\alpha) + \gamma(\alpha) D \gamma(\alpha) = P(\alpha)$ only flatly, 
i.e. $\overline{\gamma}-\gamma \in \F$.
\end{Proposition}
\begin{proof}
Let again $\overline{\gamma}-\gamma=h_{0}+h_{1}$ be the decomposition as in Claim \ref{FLATP} and 
$\gamma_{0} = \pi_{\A_{0}}(\gamma)$ the non-flat part of $\gamma$. Then follows 
$\pi_{\A_{0}}(\overline{\gamma})= \gamma_{0} + h_{0}$. Purging both ODEs of all flat contributions by applying $\pi_{\A_{0}}$
yields 
\begin{equation}
\begin{split}
\gamma_{0} + \gamma_{0} D \gamma_{0}  =P \hs{1} \mbox{and} \hs{1}
(\gamma_{0} + h_{0}) + (\gamma_{0} + h_{0})D (\gamma_{0} + h_{0}) = P,  
\end{split}
\end{equation}
where $\pi_{\A_{0}}(P) = P$.
Then we have $h_{0} \in \F$ by Proposition \ref{KYB}. Since $h_{0} \in \A_{0}$ by definition, we conclude $h_{0}=0$.  
\end{proof}

\section{Landau pole avoidance}\label{secLanAvoi}

If we insert $P(x)=x$ into the integral of (\ref{YLanCrit}), we see that $\mathcal{L}(P)<\infty$ is satisfied which, according 
to \cite{BKUY09}, means that the toy model has a Landau pole. To see this more explicitly, consider   
the RG equation for the running coupling $\alpha(L)$ 
\begin{equation}\label{RGbeta}
 \partial_{L}\alpha(L)=\beta(\alpha(L)), 
\end{equation}
which can be integrated to give
\begin{equation}\label{runC}
 L-L_{0} = \int^{\alpha(L)}_{\alpha_{0}}\frac{dx}{\beta(x)} = \ln\left|\frac{W(\xi e^{-1/\alpha(L)})}{W(\xi e^{-1/\alpha_{0}})}\right| 
\end{equation}
where $\alpha(L_{0})=\alpha_{0}$ is some reference coupling.
We find that our model has a Landau pole since the integral
\begin{equation}\label{beInt}
 \int^{\infty}_{\alpha_{0}}\frac{dx}{\beta(x)} = \ln\left|\frac{W(\xi)}{W(\xi e^{-1/\alpha_{0}})}\right| 
= \frac{1}{\alpha_{0}} + W(\xi e^{-1/\alpha_{0}}) - W(\xi)
\end{equation}
exists for any $\alpha_{0}>0$ and $\alpha(L)$ diverges for a finite $L>0$ when $L-L_{0}$ equals the value of the integral 
in (\ref{beInt}). To avoid a Landau pole, we require that this very integral diverge which in our case means that the $\beta$-function
must not grow too rapidly. For the separatrix choice $\xi = \xi_{*} = -e^{-1}$ we have 
\begin{equation}
 \beta(\alpha) \sim \sqrt{2} \alpha^{\frac{3}{2}} - \frac{2}{3} \alpha + \mathcal{O}(1)  \hs{1} \mbox{as} \hs{0.5}  
\alpha \rightarrow \infty 
\end{equation}
and thus a decreased growth compared to the instanton-free 1-loop $\beta$-function given by $\beta(\alpha)|_{\xi=0} = \alpha^{2}$ 
which is because the \emph{instantonic contribution} \emph{works towards the avoidance of a Landau pole} by means of 
the asymptotics
\begin{equation}
 1+W(-e^{-1-\frac{1}{\alpha}}) \sim \sqrt{\frac{2}{\alpha}} -\frac{2}{3 \alpha}  + \mathcal{O}(\alpha^{-2/3})   
\hs{1} \mbox{as} \hs{0.5}  \alpha \rightarrow \infty .
\end{equation}
This is an example in which the instantonic contribution alters the convergence behaviour of the integral 
\begin{equation}
 \int_{x_{0}}^{\infty} \frac{dx}{\beta(x)} 
\end{equation}
and may thus in other cases exclude the existence of a Landau pole, notwithstanding that any perturbative series of the $\beta$-function 
is blind to such contributions. 

Given the above facts about our ODE and the prominent role of the flat algebra $\F$, it is not unreasonable to assume that the 
anomalous dimension $\gamma(\alpha)$ and hence the $\beta$-function 
\begin{equation}
\beta(\alpha)=\alpha \gamma(\alpha)=\beta_{0}(\alpha) + \beta_{1}(\alpha) 
\end{equation}
has a flat piece $\pi_{\F}(\beta)=\beta_{1}$.
However, let us now be really bold and assume that this part takes the form 
\begin{equation}
 \beta_{1}(\alpha) = (\bar{\beta}(\alpha) - \beta_{0}(\alpha)) e^{-\frac{r}{\alpha}},
\end{equation}
where $r>0$ is some positive real number and $\bar{\beta}(\alpha)$ is such that 
$\int_{R}^{\infty} \bar{\beta}(x)^{-1}dx  = \infty$ for any $R >0$. Let us furthermore assume that the non-flat piece 
$\beta_{0}=\pi_{\A_{0}}(\beta)$ satisfies $\lim_{t \downarrow 0} (1-e^{-rt})\beta_{0}(1/t)=0$. Then, on account of 
\begin{equation}
 \beta(\alpha) = \beta_{0}(\alpha) + (\bar{\beta}(\alpha) - \beta_{0}(\alpha)) e^{-\frac{r}{\alpha}} \sim \bar{\beta}(\alpha) \hs{1}
\mbox{as} \hs{1} \alpha \rightarrow \infty
\end{equation}
one has 
\begin{equation}
 \int_{\alpha_{0}}^{\infty} \frac{dx}{\beta(\alpha)} = \infty 
\end{equation}
and thus a Landau pole-free theory. Alas, we do not know the real $\beta$-function of QED and it is an inherent 
feature of perturbation theory with respect to the coupling that this integral \emph{converges}\footnote{In the case of zeros of the
$\beta$-function choose $\alpha_{0}$ beyond them.}. Therefore, any approximation
of the running coupling as a solution of the RG equation (\ref{RGbeta}) within this framework is bound to have a pole which, 
however, we hardly need to remind the reader, leaves the question of a Landau pole of the real theory untouched.

\section{Landau pole of the toy model}\label{secToyLP}

We return to our toy model and solve (\ref{runC}) for the running coupling $\alpha(L)$ to find 
\begin{equation}\label{RunC}
 \alpha_{\xi}(L) = \frac{\alpha_{0}}{1-\alpha_{0}(L-L_{0}) + g_{\xi}(\alpha_{0},e^{L-L_{0}})},
\end{equation}
where $g_{\xi}(\alpha_{0},x):=\alpha_{0}(1-x)W(\xi e^{-1/\alpha_{0}})$ is flat with respect to the reference coupling 
$\alpha(L_{0})=\alpha_{0}$.
We shall now have a look at the running coupling for both spacelike and timelike photons and compare the instanton-free case $\xi=0$ with
the separatrix case $\xi=\xi_{*}(=-e^{-1})$. The hampering effect of the flat contribution on the growth of the $\beta$-function turns out
to result in lower values of the coupling in the case of large momentum transfer. This is to be expected as the $\beta$-function 
determines the momentum scale dependence of the coupling through the RG equation $\partial_{L}\alpha(L) = \beta(\alpha(L))$.  

\subsection*{Spacelike photons} 
Figure \ref{runCoup} has a plot of the running coupling $\alpha_{\xi}(L)$ with $\xi=\xi_{*}$ and the instanton-free case $\xi=0$ for spacelike 
photons where $L \in \mathbb{R}$ due to $-q^{2}>0$ and reference coupling $\alpha_{0}=0.1$.
\begin{figure}[ht] 
\includegraphics[height=7cm]{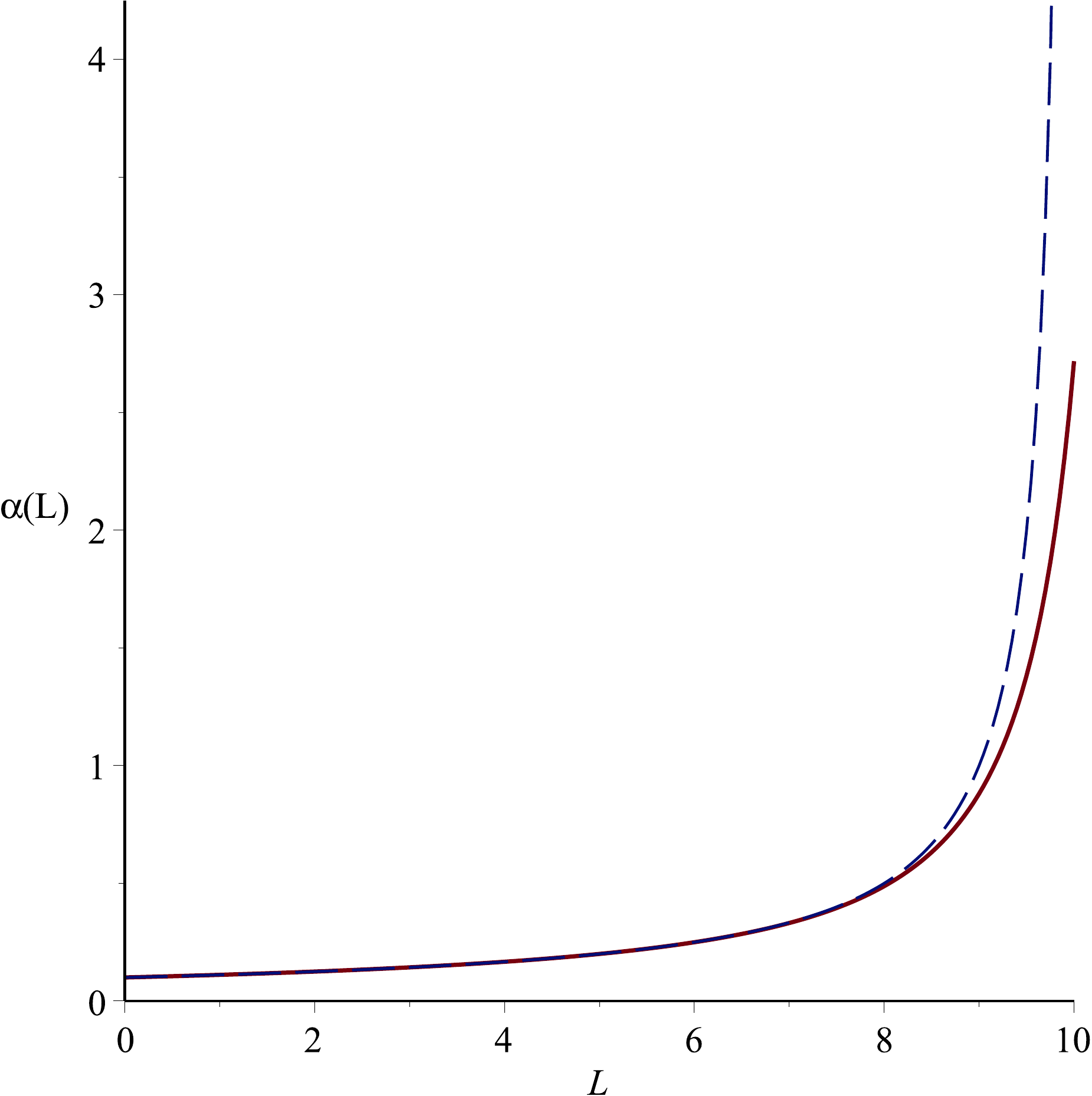}  
\caption{ \small Spacelike photons: the running coupling $\alpha_{\xi}(L)$ for $\xi=\xi_{*}$ (solid) and the instanton-free case $\xi=0$ 
(dashes) at reference coupling $\alpha_{0} = 0.1$ for $L_{0}=0$. } 
\label{runCoup} 
\end{figure}
The diagram shows that the flat piece $g_{\xi}(\alpha_{0},e^{L})$ shifts the Landau pole from $L'=1/\alpha_{0}$ to the solution of 
\begin{equation}\label{transeq}
L''=\frac{1}{\alpha_{0}} + (1 - e^{L''}) W(-e^{-1-1/\alpha_{0}}).
\end{equation}
Note that $L''>L'$ due to $(1 - e^{L}) W(-e^{-1-1/\alpha_{0}}) > 0$ for $L>0$ and that by the 
transcendentality of (\ref{transeq}) there is no canonical way to define a reference scale, usually denoted by $\Lambda$. 

\subsection*{Timelike photons}
In the case of timelike photons, when $-q^{2}<0$, the running coupling in (\ref{RunC}) has an imaginary part, we write (\ref{RunC}) 
in the form 
\begin{equation}\label{timeCoupl}
\widetilde{\alpha}_{\xi}(-s):= \alpha_{\xi}(\log(-s)) = \frac{\alpha_{0}}{1-\alpha_{0}\log(-s)+ g_{\xi}(\alpha_{0},-s)},
\end{equation}
where $L=\log(-s)$ with $s:=q^{2}/\mu^{2}$ and $L_{0}=0$ for reference point $s_{0}=-1$. Complex-valued 'timelike' couplings 
have been studied in the case of QCD: \cite{PenRo81} have argued that $|\widetilde{\alpha}(-s)|$ is to be favoured over $\widetilde{\alpha}(|-s|)$ as perturbation 
coupling parameter for timelike processes. They find better agreement with experimental results at lower order of perturbation theory. 
Whether or not this pertains to QED, Figure \ref{Fig:RunCoup} has a plot of this parameter in the two cases $\xi=\xi_{*}$ and $\xi=0$.
\begin{figure}[ht] 
\includegraphics[height=7cm]{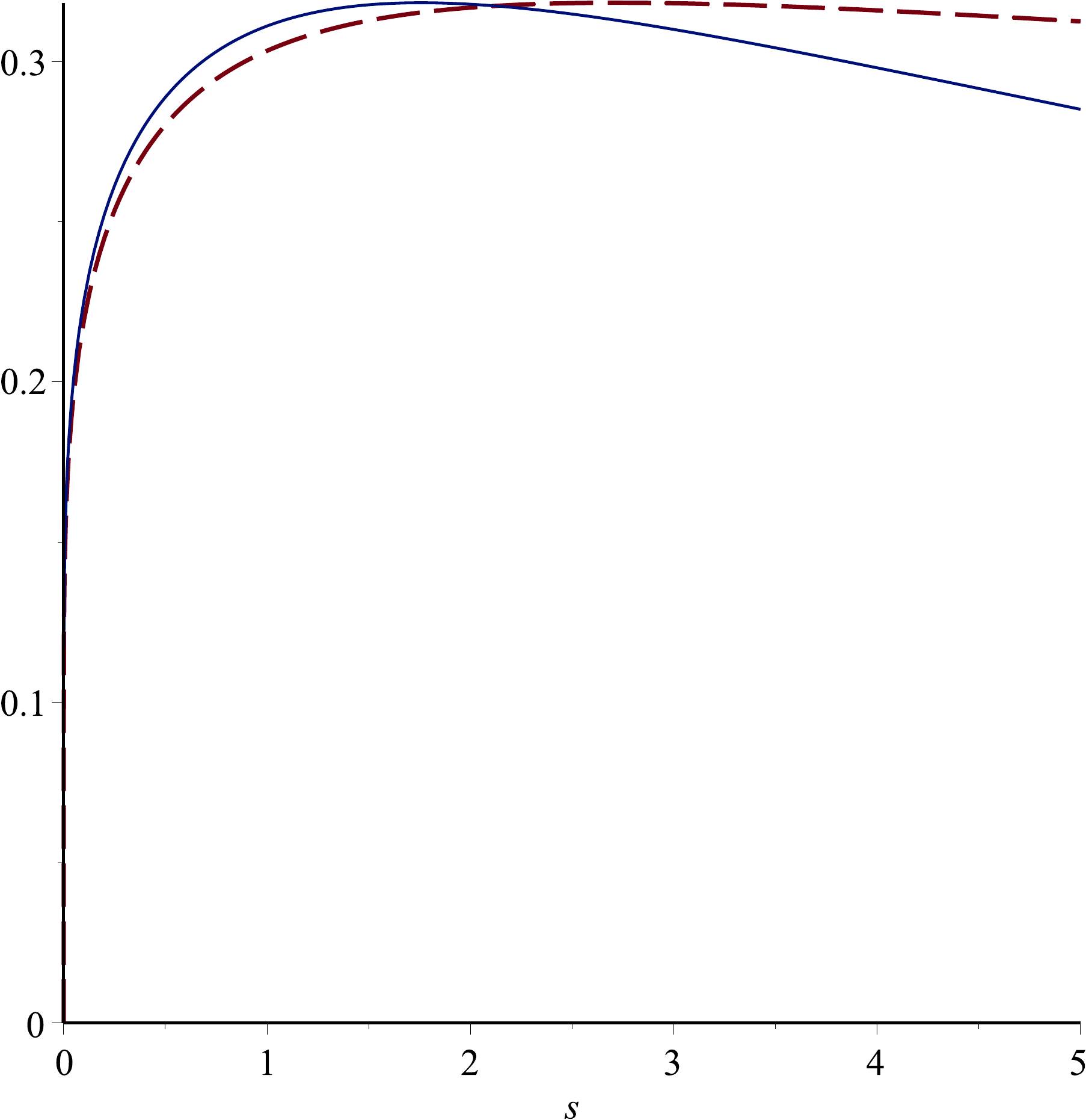}  \hs{0.2} \includegraphics[height=7cm]{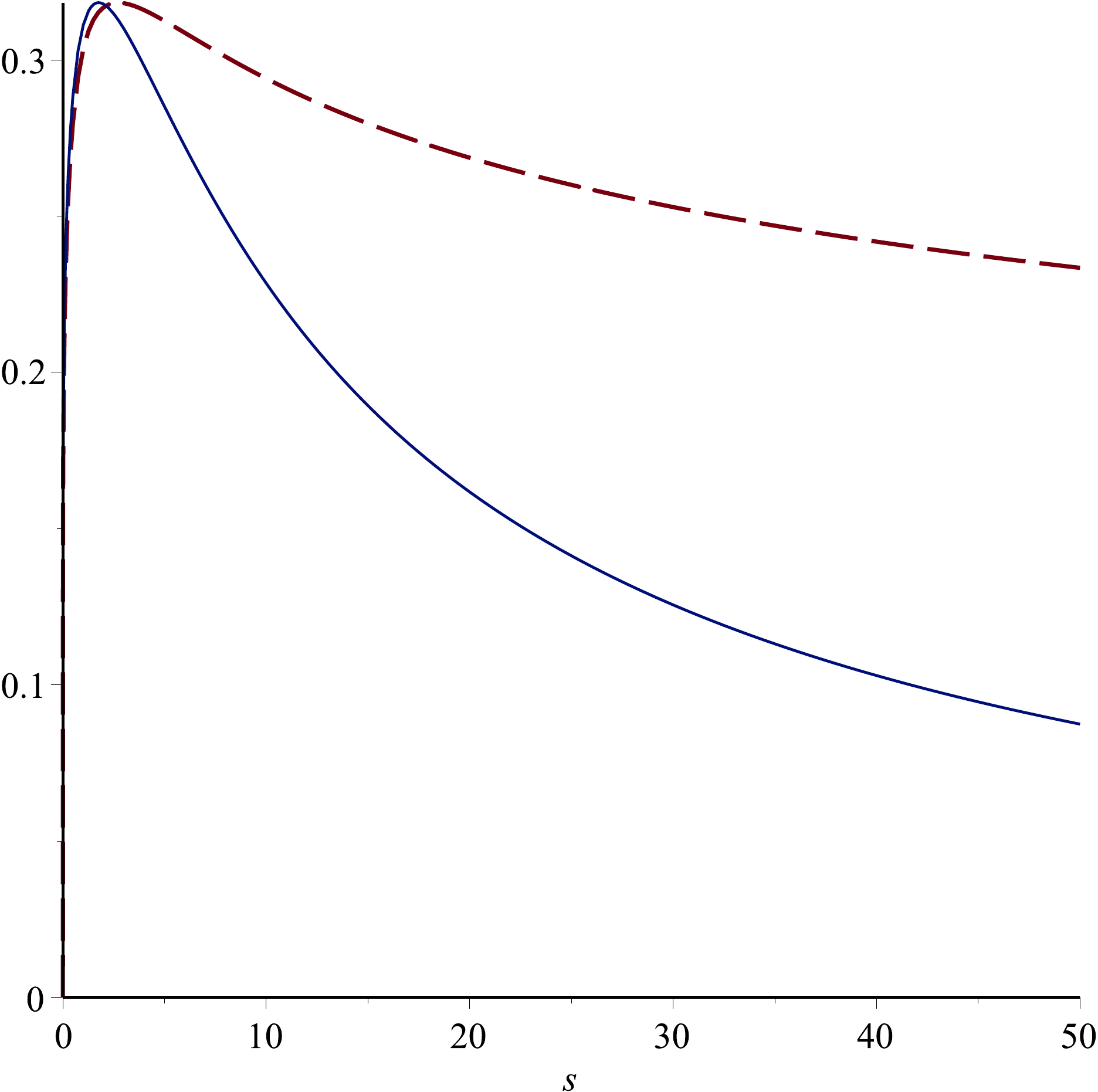}  
\caption{ \small Coupling parameter $|\widetilde{\alpha}_{\xi}(-s)|$ for timelike photons at lower (left) and larger (right) momenta: 
$\xi=\xi_{*}$ (solid) and $\xi=0$ (dashes) with reference coupling $\alpha_{0} = 1$. } 
\label{Fig:RunCoup} 
\end{figure}
Because the coupling is complex-valued, we consider the dispersion relation
\begin{equation}\label{DRInt}
\widetilde{\alpha}_{\xi}(-s)=\int_{0}^{\infty} d\eta \ \frac{\Omega_{\xi}(\eta)}{s-\eta}
\end{equation}
and calculate the spectral density by taking the limit
\begin{equation}
 \lim_{\varepsilon \downarrow 0} \left\{\widetilde{\alpha}_{\xi}(-x-i\varepsilon) - \widetilde{\alpha}_{\xi}(-x+i\varepsilon) \right\}
= - 2 \pi i \Omega_{\xi}(x).
\end{equation}
yielding
\begin{equation}
 \Omega_{\xi}(\eta) =  \frac{\alpha^{2}_{0}}{[1-\alpha_{0}\log(\eta)+ g_{\xi}(\alpha_{0},\eta)]^{2}+ (\alpha_{0}\pi)^{2} },
\end{equation}
which equals the square modulus of the timelike coupling in (\ref{timeCoupl}): $\Omega_{\xi}(\eta)=|\widetilde{\alpha}_{\xi}(-\eta)|^{2}$.
Apart from the absolute value, the spectral density has a graph of the same shape as that in Figure \ref{Fig:RunCoup}, where we see that the 
instantonic contribution shows a significant effect at strong reference coupling $\alpha_{0}=1$: beyond the pair-creation 'bump', 
higher mass state contributions are suppressed. The seemingly natural interpretation of this for 
timelike photons in terms of a weaker interaction in the s-channel where electrons and positrons annihilate has to be more than taken
with a pinch of salt: though these deviations seem to be significant, we have to remind ourselves that these
are toy model results. As pointed out in Section 3, we cannot expect our (single flavour) toy QED to hold for large couplings around
$\alpha_{0}=1$. To make the impact of the flat contribution visible, we have to go up to this level of the coupling strength
and accept that we at the same time enter the realm of a toy model:
the implicit assumption of choosing this reference coupling is that the running coupling is still given by (\ref{timeCoupl}).
 
However, both in the case of timelike and spacelike photons, 
it is by no means far-fetched to conclude that if flat contributions impede the $\beta$-function's growth, the running coupling
may exhibit lower values at higher momentum transfer also in a real-world (3-flavour) QED.

\section{Photon self-energy}\label{secSelf}

If we take the anomalous dimension in (\ref{anoDim}) setting $c=1$, apply the RG recursion (\ref{rec2}) and calculate all higher 
log-coefficients, we find for the self-energy 
\begin{equation}
 \Sigma(\alpha,L):= \gamma(\alpha)\cdot L = \sum_{j \geq 1}\gamma_{j}(\alpha) L^{j}
\end{equation}
an interesting result which we present in the next 
\begin{Claim}
Given the exact solution $\gamma(\alpha)$ of the nonlinear toy model ODE (\ref{ode}), the recursion (\ref{rec2}) yields
\begin{equation}
 j!\gamma_{j}(\alpha)= \alpha W(\xi e^{-\frac{1}{\alpha}}), \hs{2} j \geq 2
\end{equation}
and 
\begin{equation}
 \Sigma_{\xi}(\alpha,e^{L}) =  \alpha L +  \alpha (e^{L}-1) W(\xi e^{-\frac{1}{\alpha}})
= \alpha \ln(-q^{2}/\mu^{2}) - \alpha \left(1+q^{2}/\mu^{2}\right) W(\xi e^{-\frac{1}{\alpha}})
\end{equation}
for the photon self-energy.
\end{Claim}
\begin{proof}
We proceed by induction. First $j=2$. 
\begin{equation}
 2!\gamma_{2}(\alpha)= \gamma(\alpha)(\alpha \partial_{\alpha}-1)\gamma(\alpha) \stackrel{(\ref{ode})}{=} 
\gamma(\alpha) - \alpha = \alpha W(\xi e^{-\frac{1}{\alpha}}).
\end{equation}
Let now $j\geq 2$. Then
\begin{equation}
\begin{split}
(j+1)!\gamma_{j+1} &\stackrel{(\ref{rec2})}{=} \gamma(\alpha \partial_{\alpha}-1)j!\gamma_{j} 
= \gamma(\alpha \partial_{\alpha}-1)\alpha W 
= \gamma(\alpha \partial_{\alpha}-1)(\gamma - c\alpha) \\
&=  \gamma(\alpha \partial_{\alpha}-1)\gamma  
= 2 \gamma_{2} = \alpha W ,
\end{split}
\end{equation}
where we have used $(\alpha\partial_{\alpha}-1)\alpha=0$ in the fourth step. For the self-energy then follows 
\begin{equation}
 \Sigma_{\xi}(\alpha,e^{L}) = \gamma \cdot L = \gamma L + \sum_{j\geq 2} \gamma_{j}L^{j} = \alpha (1+W) L 
+ \sum_{j\geq 2} \frac{1}{j!} \alpha W L^{j} = \alpha L +  \alpha (e^{L}-1) W 
\end{equation}
and thus the result. 
\end{proof}
In the notation of the previous section we set $s=q^{2}/\mu^{2}$ with Minkowski momentum $q \in \mathbb{R}^{1,3}$ and write
\begin{equation}
 \Sigma_{\xi}(\alpha,-s) =  \alpha \log(-s) - g_{\xi}(\alpha,-s).
\end{equation}
To see how the instantonic contribution affects the renormalized propagator, we define 
\begin{equation}
 \Pi_{\xi}(\alpha,-s):=\frac{1}{s[1 - \Sigma_{\xi}(\alpha,-s)]} = \frac{1}{s[1-\alpha \log(-s) + g_{\xi}(\alpha,-s)]}
\end{equation}
and study its properties for $\xi=\xi_{*}$ and in particular how the flat contribution causes this quantity to deviate from its 
instanton-free version $\Pi_{0}(\alpha,-s)=\Pi_{\xi}(\alpha,-s)|_{\xi=0}$.

\subsection*{Spacelike photons} For spacelike photons, the Green's function is real-valued due to $S:=-s>0$ and leads to
the propagator 
\begin{equation}
\Pi_{\xi}(\alpha,S)= -\frac{1}{S[1-\alpha \log(S) + g_{\xi}(\alpha,S)]}.
\end{equation}
To see the flat contribution's impact, we compare this quantity for $\xi=\xi_{*}$ with $\Pi_{0}(\alpha,S)$. 
\begin{figure}[ht] 
\includegraphics[height=7cm]{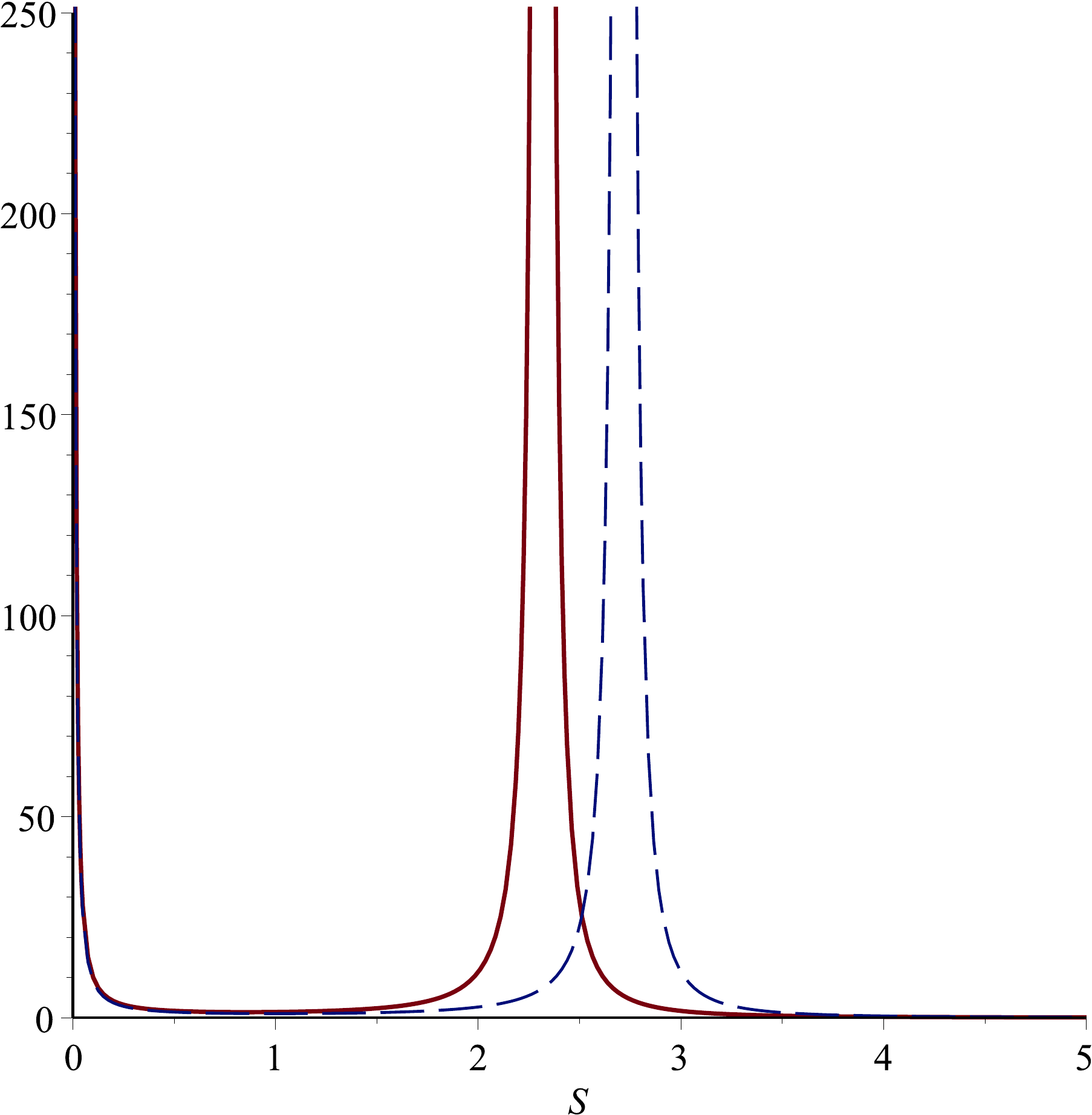}  
\includegraphics[height=7cm]{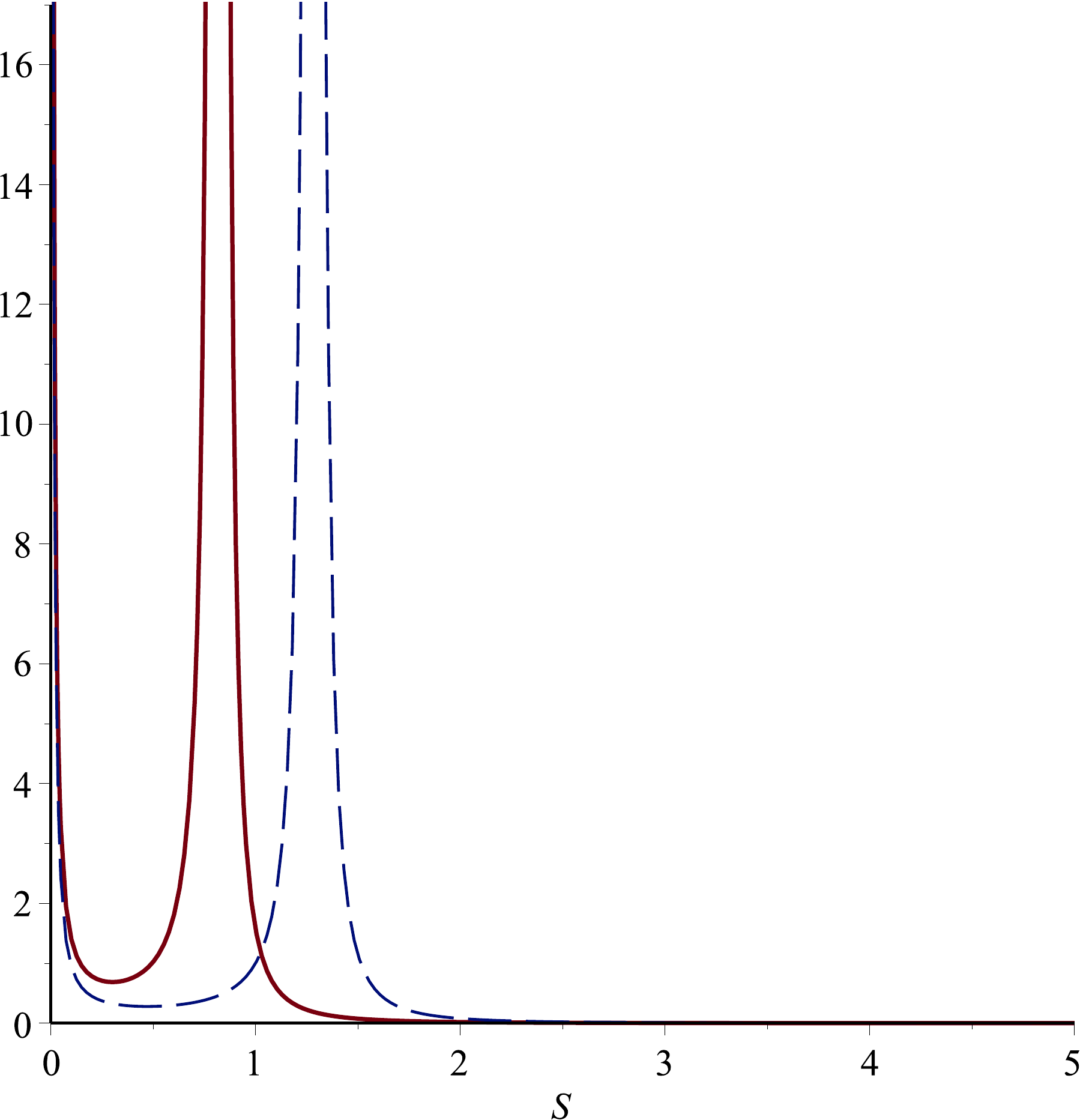}  
\caption{ \small Pole shift of the propagator squares $|\Pi_{0}(\alpha,S)|^{2}$ (dashed line), $|\Pi_{\xi}(\alpha,S)|^{2}$ 
(solid line) for $\xi=\xi_{*}$ and spacelike photons at $\alpha=1$ (left) and $\alpha=4$ (right). } 
\label{Reson} 
\end{figure}
Figure \ref{Reson} has plots of the squares $|\Pi_{0}(\alpha,S)|^{2}$ and $|\Pi_{\xi}(\alpha,S)|^{2}$ displaying two aspects: 
firstly, the propagator exhibits 
a pole which is situated at higher momenta in the weak coupling regime than in the strong coupling regime. 
Secondly, the instantonic contribution causes a \emph{pole shift} towards lower momenta, where this effect is more
pronounced at larger and negligible at lower values of the coupling.      
However, since these poles are those of a toy model, we refrain from any interpretation.

\subsection*{Timelike photons: K\"allén-Lehmann spectral function}\label{Kaleh}

For timelike photons, where $-q^{2}<0$ and thus $-s<0$, we ask for the spectral function $\rho_{\xi}(\alpha,\omega)$ in the 
K\"allén-Lehmann spectral form of the propagator  
\begin{equation}\label{SpecForm}
 \Pi_{\xi}(\alpha,-s)=\frac{1}{s[1 - \Sigma_{\xi}(\alpha,-s)]} = \frac{1}{s} + 
\int_{0}^{\infty} d\omega \   \left(\frac{1}{s-\omega}-\frac{1}{1-\omega} \right)  \rho_{\xi}(\alpha,\omega),
\end{equation}
where the integrand has been chosen so as to warrant the renormalization condition $s\Pi(\alpha,-s)|_{s=1}=1$. 
To extract the spectral function, we compute the limit 
\begin{equation}
 \lim_{\varepsilon \downarrow 0} \left\{\Pi_{\xi}(\alpha,-x-i\varepsilon) - \Pi_{\xi}(\alpha,-x+i\varepsilon) \right\}
= - 2 \pi i \rho_{\xi}(\alpha,x)
\end{equation}
for $x>0$ and obtain the K\"allén-Lehmann spectral function 
\begin{equation}\label{KaellLeh}
\rho_{\xi}(\alpha,\omega) =  \frac{\alpha }{\omega }
 \frac{1}{[1- \alpha \ln\omega +  g_{\xi}(\alpha,-\omega)]^{2} + (\alpha \pi)^{2} }.
\end{equation}
Figure \ref{KLS} shows a plot of the spectral function $\rho(\alpha,\omega):=\rho_{\xi_{*}}(\alpha,\omega)$ for the separatrix solution at 
different coupling strengths $\alpha$.  
\begin{figure}[ht]
\begin{center} \includegraphics[height=8cm]{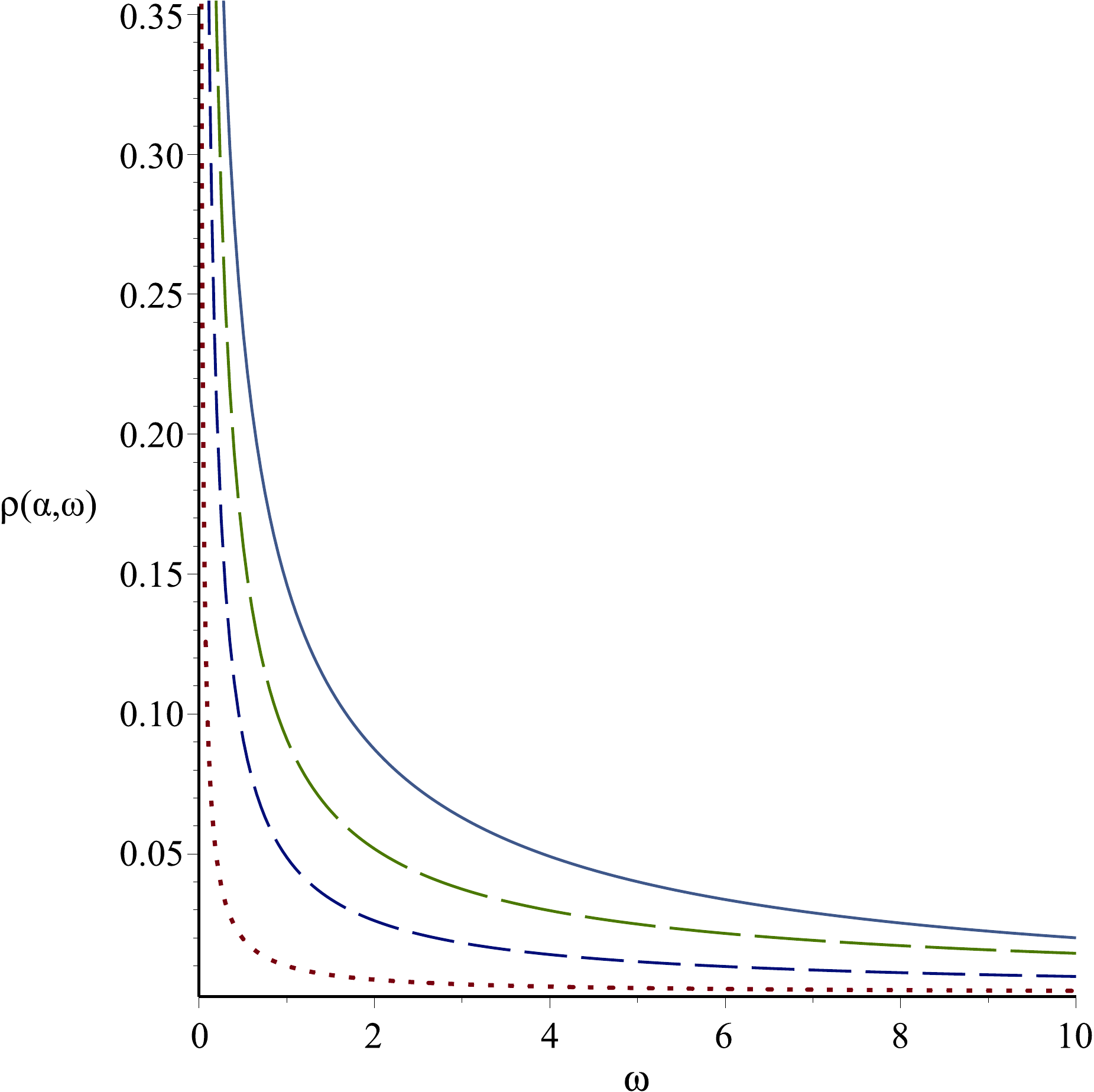}  \end{center} 
\caption{\small Spectral function for various coupling strengths: $\alpha=0.01$(dots), $\alpha=0.05$(dashes), $\alpha=0.1$(long dashes) 
and $\alpha=0.5$(solid). }
\label{KLS}
\end{figure}
Notice that on account of the primitive
\begin{equation}\label{primi}
\int \frac{d\omega}{\omega \ln^{2}\omega} = -\frac{1}{\ln \omega} 
\end{equation}
the dispersion integral in (\ref{SpecForm}) has no trouble converging for $s\neq 0$, at neither integration bound.  
However, apart from the fact that we are dealing with a toy model here, we have considered \emph{massless} QED, 
and can thus not expect our spectral function to encapsulate any valid physics below the pair-creation threshold 
$\omega_{0}\approx 4m^{2}$.  

\subsection*{Instantonic contribution} To again see the alterations brought about by the non-perturbative contribution, we compare
the spectral function $\rho(\alpha,\omega)$ with its instanton-free version $\rho_{0}(\alpha,\omega)$. We find that the flat part 
\begin{equation}
 g_{\xi}(\alpha,-\omega)= \alpha (1+\omega) W(\xi e^{-\frac{1}{\alpha}}) 
\end{equation}
does only play a role at large enough couplings. As the diagrams of Figure \ref{rho} show for $\alpha=5$, the flat contribution leads to 
a slightly increased contribution of lower mass states. For large masses $\omega$ there is only a small change towards a smaller 
contribution.   

\begin{figure}[ht] 
\includegraphics[height=7cm]{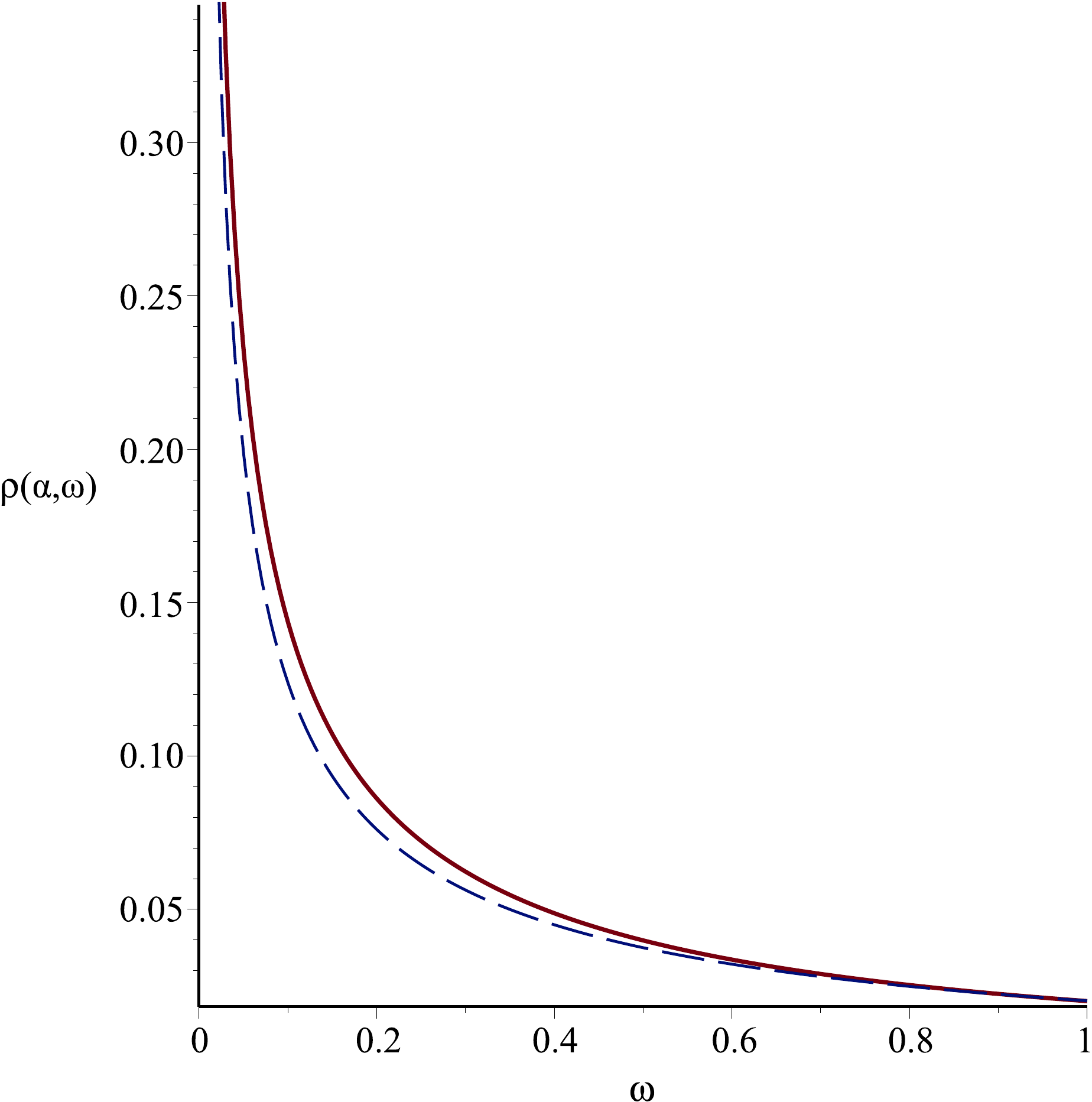}  
\includegraphics[height=7cm]{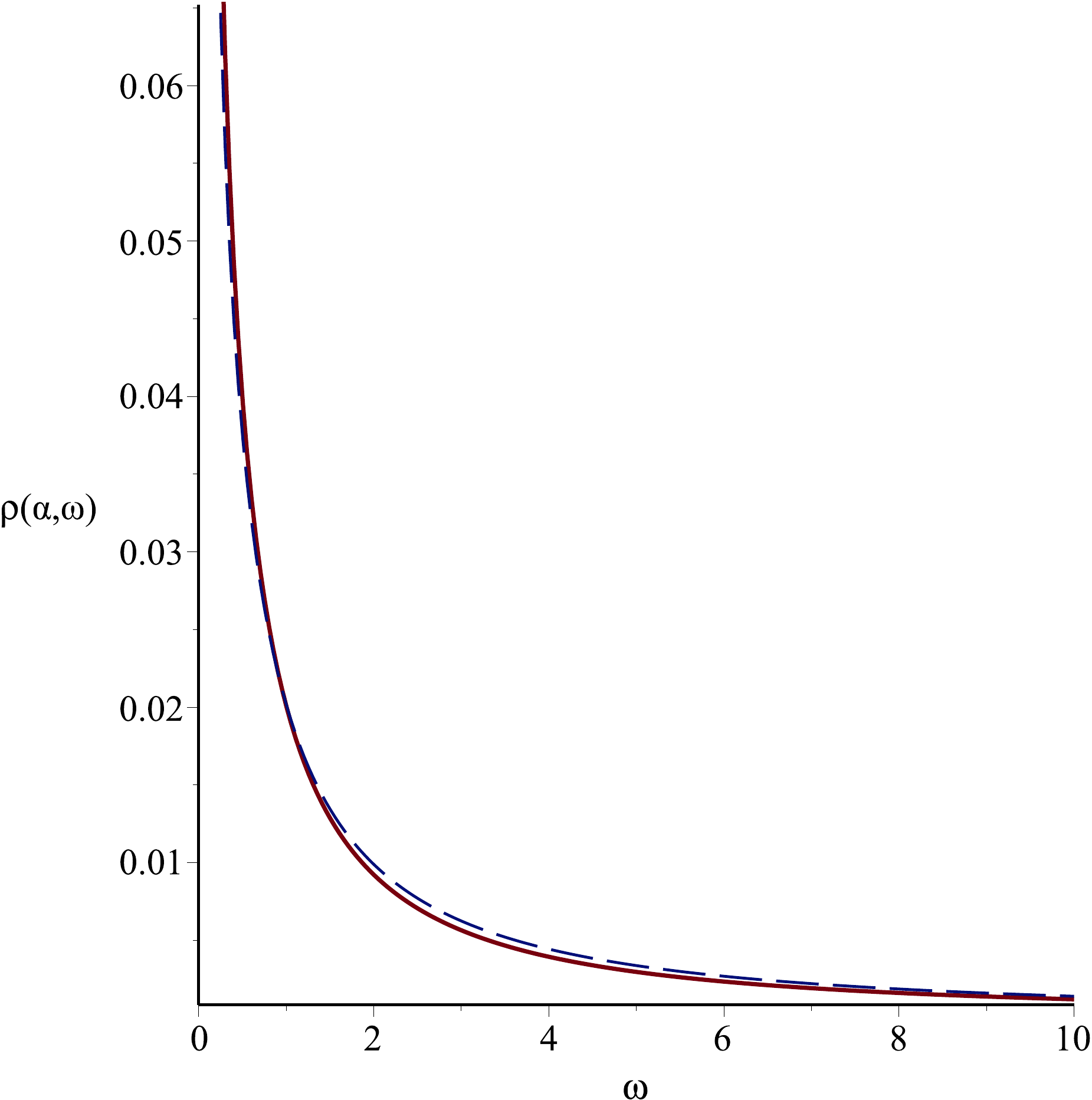}  
\caption{ \small The two spectral functions $\rho_{0}(\alpha,\omega)$ (dashed line), $\rho(\alpha,\omega)$ (solid line) for 
$\alpha=5$ and small/large mass contributions (diagram on the left/right). } 
\label{rho} \end{figure}

However, for large masses $\omega$ the function $g_{\xi}(\alpha,-\omega)$ dominates over the logarithmic 
part in the denominator of $\rho(\alpha,\omega)$ in (\ref{KaellLeh}) and suppresses higher mass contributions much more than the 
logarithmic contribution by itself.   
Interestingly enough, due to the fact that $g_{\xi}(\alpha,-\omega)$ will dominate over any polynomial in the variable $L=\ln \omega$ 
for large $\omega$, this picture would not change qualitatively if we took higher loop contributions into account.

\section{Conclusion}

We have considered a non-standard perturbative approach to approximate the QED $\beta$-function and the photon propagator in which
QED is reduced to a single non-linear ODE for the anomalous dimension. This differential equation proved to habour a sufficient 
criterion for the existence of a Landau pole and thus a necessary condition for the possibility of QED to be free of such ailment. 
A non-perturbative contribution of the instantonic type, 
emerging already at first loop order, proved to be capable of hampering the growth of the $\beta$-function. It is important to 
note that such contributions may determine the asymptotic behaviour of the $\beta$-function so as to possibly exclude 
the existence of a Landau pole. Investigating the impact of the instantonic contribution on both the running coupling and the Green's
function, we have found deviations from the standard instanton-free solutions at larger couplings.   

\section{Acknowledgments}

D.K. is supported by an Alexander von Humboldt Professorship from the Alexander von Humboldt Foundation and the BMBF. D.K. and L.K. thank
David Broadhurst and John Gracey for interesting and helpful discussions. L.K. would like to express his gratitude to
Karen Yeats for explaining and clarifying a number of issues. He owes special thanks to Erik Panzer for extensive discussions and 
illuminating remarks.

\end{document}